%% file: main.tex
\documentclass[twoside,11pt]{article}
\usepackage{jmlr2e}

\usepackage{caption}
\usepackage{color}
\usepackage{hyperref}
\hypersetup{
    colorlinks = true, 
    linkcolor = blue,
    filecolor = magenta, 
    citecolor = green,
    urlcolor = cyan,
    pdfborder = {0 0 0}
}
\usepackage{algorithm}
\usepackage[noend]{algorithmic}
\newtheorem{theorem}{Theorem} 
\newtheorem{proposition}{Proposition}
\newtheorem{lemma}{Lemma}

\newtheorem{definition}{Definition}

\newtheorem{remark}{Remark}
\usepackage{booktabs}
\usepackage{amsmath} 
\usepackage{enumitem}
\usepackage{bm}

\input{def}

\jmlrheading{1}{2021}{1-48}{4/00}{10/00}{meila00a}{Fan Yao, Chuanhao Li, Denis Nekipelov, Hongning Wang and Haifeng Xu}

\ShortHeadings{Human-vs-GenAI Content Creation Competition}{Human-vs-GenAI Content Creation Competition}
\firstpageno{1}

\begin{document}

\title{Human vs. Generative AI in Content Creation Competition: \\ Symbiosis or Conflict?}

\author{\name Fan Yao$^1$ \email fy4bc@virginia.edu
       \AND
       \name Chuanhao Li$^2$ \email chuanhao.li.cl2637@yale.edu 
    %   , Charlottesville, VA 22904, USA
       \AND
       \name Denis Nekipelov$^{1,3}$ \email dn4w@virginia.edu 
       \AND
       \name Hongning Wang$^{4}$ \email wang.hongn@gmail.com
       \AND
       \name Haifeng Xu$^{5}$ \email haifengxu@uchicago.edu \\ \\ 
       \addr $^1$Department of Computer Science, University of Virginia, USA \\
       \addr $^2$Department of Statistics and Data Science, Yale University, USA\\
       \addr $^3$Department of Economics, University of Virginia, USA\\
       \addr $^4$Department of Computer Science and Technology, Tsinghua University, China\\
       \addr $^5$Department of Computer Science, University of Chicago, USA
    }

\maketitle

\begin{abstract} 
%\hf{Is Hongning's affiliation correct? Also, maybe switch to a different template if we do not plan to publish at JMLR. I realized that our previous template caused some confusions among readers as some people  thought that our paper was at JMLR, which is not true.  Any plan for which journal to submit? I have template for OR papers if needed, otherwise a plain template will also cause less confusion. }
The advent of generative AI (GenAI) technology produces transformative impact on the content creation landscape, offering alternative approaches to produce diverse, high-quality content across media, thereby reshaping online ecosystems but also raising concerns about market over-saturation and the potential marginalization of human creativity.
Our work introduces a competition model generalized from the Tullock contest to analyze the tension between human creators and GenAI. Our theory and simulations suggest that despite challenges, a stable equilibrium between human and AI-generated content is possible. Our work contributes to understanding the competitive dynamics in the content creation industry, offering insights into the future interplay between human creativity and technological advancements in GenAI.
\end{abstract}

\input{intro}

\input{model}

\input{theory}

\input{experiment}

\input{conclusion}

% \section{Acknowledgement}
% This work is supported in part by the US National Science Foundation under grants IIS-2007492, IIS-2128019 and IIS-1838615. Haifeng Xu is supported in part by an ARO award W911NF-23-1-0030. 

\vskip 0.2in

\bibliography{main}

\appendix 
\input{appendix}

\end{document}

%% file: def.tex
\newcommand{\ingamed}{{$\G^{(1)}_{IN}$}}
\newcommand{\basegame}{{$\G^{(1)}_{EX}$}}
\newcommand{\ingame}{{$\G_{IN}$}}
\newcommand{\exgame}{{$\G_{EX}$}}

\def\min{\qopname\relax n{min}}
\def\max2{\qopname\relax n{max2}}
\def\max{\qopname\relax n{max}}

\newcommand{\RR}{\mathbb{R}}

\def\G{\mathcal{G}}

\def\X{\mathcal{X}}
\def\Y{\mathcal{Y}}

\def\tt{\bm{\theta}} 
\def\x{\bm{x}} 
 
\def\y{\bm{y}} 
\def\w{\mathbf{w}} 
 
\def\c{\bm{c}}

\def\g{\bm{g}} 
 
\def\s{\bm{s}}

\def\w{\bm{w}}

\def\U{\mathcal{U}}

\newenvironment{lp*}{\begin{equation*}  \begin{array}{lll}}{\end{array}\end{equation*}}

%% file: intro.tex
\section{Introduction}

In January 2024, a novel written by making full use of ChatGPT, ``The Tokyo Tower of Sympathy'', won the most prestigious Japanese literary award Akutagawa Prize \citep{cnn2024}. This marks the era of content creation competition between humans and generative AI (GenAI).

%The advent of generative AI (GenAI) technologies, such as ChatGPT \cite{brown2020language}, Midjourney and Runway \cite{ho2020denoising}, has revolutionized not only individual lives but also entire society, particularly in the realm of online content creation industry. 
In addition to the impressive text generation abilities demonstrated by ChatGPT \citep{brown2020language}, the generation of high-quality multi-media content, such as image and video \citep{ho2020denoising}, has also been made widely accessible to ordinary users by a spectrum of easy-to-use tools (such as Midjourney and Sora). This has revolutionized not only individual lives but also the entire society, particularly in the realm of online content sharing industry.
These AI-driven solutions have significantly altered the landscape of content creation on popular online platforms like Instagram and Tiktok, which are deeply woven into the fabric of our daily digital experience. Specifically, GenAI lowers entry barrier for individuals who may lack the technical skills or resources to produce high-quality content, and may even help them gain a competitive edge in the market. This shift has led to a surge in the volume of content being created and shared, fostering new dynamics in online content ecosystems \citep{epstein2023art,wahid2023written}.

However, this AI-driven transformation also brings significant challenges and concerns \citep{wach2023dark}. The ease of generating content with AI can potentially lead to an over-saturated market, making it harder for individual creators to stand out or leaving truly creative human creators under appreciated \citep{doshi2023generative}. A most recent event is that Universal Music Group pulls songs from TikTok and accuses the platform of being ``flooded with AI-generated recordings'' that diluted the royalty pool for real, human musicians \citep{nytimes2024}. This echos the Gresham's Law that ``bad money drives out good'' \citep{selgin2020gresham}. 

On the other hand, GenAI models are not omniscient. A key limitation---or perhaps a defining characteristic---of these models is their dependency on extensive and diverse datasets of high-quality, human-generated content for training \citep{bertrand2023stability,briesch2023large}. Should GenAIs inadvertently marginalize productive, high-quality human content creators, the resultant decline in the quality of model-generated content is inevitable. Drawing an analogy to biological interactions, the dynamics between GenAI-based creators and human creators could evolve into either \emph{symbiosis}, leading to a mutually beneficial equilibrium, or \emph{antagonistic conflict}, perpetuating rivalry and potentially destabilizing the market.  Therefore, an urgent and scientifically interesting question to ask is, whether human creators will be driven out of the market when competing against AI-generated content, or is there a path toward a stable, symbiotic relationship?

In this paper, we propose a stylized model to depict the rivalry between traditional human content creators and those utilizing GenAI technology. Our framework expands on the Tullock contest model \citep{tullock2001efficient}, a model extensively applied in econometrics to analyze competitive scenarios. We first explore the impact of GenAI as an external influencer on the equilibrium state of human creators, and then delve into scenarios where creators have the autonomy to either adopt or refrain from using GenAI tools when making their content. Our theoretical results and empirical findings deliver a promising outlook: despite GenAI's potential to disrupt the human content generation market, a stable equilibrium with desirable characteristics is attainable. Our contributions lie in the following three aspects:
\begin{enumerate}
    \item \textbf{modeling-wise}, we are the first to formally propose a mathematical model to characterize the competition between human and GenAI, including an exclusive competition context where GenAI acts as an exogenous source and an inclusive context where each creator can strategically decide to use GenAI or not; 
    \item \textbf{conceptually}, our theories and experiments provide encouraging answers to important questions regarding human-vs-AI competition,  derive new insights and offer prediction about  future online content markets in the coming era; 
    \item \textbf{technique-wise}, our model not only generalizes the classic Tullock contest but also derives several novel attributes concerning the equilibrium of the competition, enriching the existing literature with new theoretical advancements.
\end{enumerate}

\section{Related Work}

Recently, the strategic interaction and the resulting dynamics among viewers, creators, and the recommender system have attracted considerable research attention \citep{dean2024recommender,acharya2024producers,ben2020content,yao2022learning,agarwal2022diversified,boutilier2023modeling,yao2022learning2,dean2022preference}.
Among these, an emerging line of work focuses on online content economy and the modeling of content creator competitions \citep{ben2017shapley,ben2018game,yao2023rethinking,yao2023bad,zhu2023online,hu2023incentivizing,jagadeesan2023supply,hron2022modeling}. In these models, content creators strategically choose their production strategies, e.g., the quality \citep{hu2023incentivizing,hron2022modeling} or type \citep{jagadeesan2022supply} of their content, and compete for different objectives such as traffic \citep{hron2022modeling,ben2017shapley}, user engagement \citep{yao2023bad}, or platform provided incentives \citep{zhu2023online,yao2023rethinking}. Some of them aim to understand the property of creator side equilibrium, for example, how creators will specialize at the equilibrium \citep{jagadeesan2022supply}, how creators' strategic behavior affects social welfare \citep{yao2023bad}, and how to design optimization method for long-term welfare considering content creators' strategic behaviors \citep{ben2017shapley,ben2018game,yao2023rethinking,zhu2023online,hu2023incentivizing,immorlica2024clickbait,mladenov2020optimizing}. Our competition model introduces GenAI creators into the arena for the first time and we investigate the impact of GenAI technology to human creators through analyzing the properties of the competition equilibrium. 

In Tullock contest \citep{tullock2001efficient}, also known as lottery contest, the probability of each player winning a fixed prize is the ratio between the effort she spends and the total effort exerted by all players. The Nash equilibrium of one-dimensional Tullock contest with homogeneous cost is well understood \citep{ewerhart2015mixed, EWERHART2017168} and some natural extensions have been well studied, for example, the prize value is a linear function \citep{chowdhury2011generalized}, players are equipped with convex loss \citep{ghosh2023best}. Recent works employ game-theoretical models similar to Tullock contest \citep{hron2022modeling,yao2023bad} to model content creator competition. Our model extends the scope of Tullock contest by introducing GenAI players and properties of such players based on the up-to-date understandings of foundation models behind such technology, and considering heterogeneous cost functions.

%% file: model.tex
\section{Modeling  Content Creation Competition between Humans and GenAI}
In this section, we formally introduce our model for human and GenAI content creation competition. Our model is rooted in and strictly generalizes the textbook model of the \emph{Tullock contest} \citep{tullock2001efficient}, which is perhaps the most widely adopted paradigm to model contests \citep{dechenaux2015survey,szymanski2003economic,mueller2003public} and has recently been used to model competitions among content creators \citep{hron2022modeling,yao2023bad} and bitcoin miners \citep{leshno2020bitcoin,arnosti2022bitcoin}.            

\vspace{2mm}
\noindent
\textbf{Modeling human creators. } There are $n$ human content creators competing over $K$ topics. Throughout, we use the notation $[n] = \{ 1, 2, \cdots, n \}$ for any integer $n$. In practice, each topic $k \in [K]$ can be viewed as either an explicit subject (e.g., a trending tag) or a latent theme associated with a user preference group. Each creator competes for user attention by generating content for different topics. Formally, let  $x_{ik} \in [0, \infty)$ denote the calibrated body of content on topic $k$ generated by creator $i$, where calibration accounts for both the quality and quantity of content. In words, $x_{ik}$ captures creator $i$'s level of competitiveness on topic $k$, and is referred to as the \emph{body of content}.  
%\footnote{We remark that $x_{ik}$ should not   be interpreted simply as the ``number'' of contents; its value should be calibrated so that    quality of the contents are accounted for as well.}  
Following game-theoretic conventions, we refer to the (deterministic) effort allocation vector  $\x_i=(x_{ik})_{k=1}^K\in \RR_{\geq 0}^K$  as a \emph{pure strategy} of   creator $i \in [n]$. Naturally, content creation is costly, and we use $c_i(\x_i)$ to denote creator $i$'s cost resulting from her effort allocation $\x_i$. Throughout the paper, we assume the cost functions to be convex and twice-differentiable, as widely adopted in recent literature for modeling content creation competition \citep{jagadeesan2023supply,yao2023rethinking}, previous literature on contest modeling \citep{szymanski2003economic,mueller2003public} and in general models of production by firms   \citep{shephard2015theory}. This assumption captures the feature that it can be easy to create some content, but continuously creating high-quality content becomes significantly more costly.      
%where $x_{ik}$ describes the amount of effort spent on topic $k \in [K]$. To be more concrete, the ``effort'' value $x_{ik}$ can be regarded as the quality-calibrated body of content generated by creator $i$ on topic $k$. 

\vspace{2mm}
\noindent
\textbf{GenAIs as a new type of  content creators. } The core novelty introduced by our modeling, both conceptually and technically, is the integration of the GenAI
% \emph{generative AI} (GenAI) 
into the content creation competition. The competitiveness of GenAI on each topic $k$ depends on two major factors: (a) the \emph{total} body of content that human creators have generated for topic $k$, i.e., the source of training data for GenAI; and (b) the learning capability of the GenAI model. Formally, given all human creators'  strategies $\{ \x_i \}_{i=1}^n$, the \emph{total} body of created content on topic $k$ is   $s_k =\sum_{j=1}^n x_{jk}$. We assume that the calibrated body of content that GenAI creates for topic $k$ can be  described by function 
\begin{equation}
\label{eq:AIcontent}
    g_k(s_k) = \alpha_k \cdot \big(s_k \big)^{\tilde{\beta}_k},
\end{equation}
where $\alpha_k$ captures the \emph{efficiency of data usage} by the GenAI model whereas $\tilde{\beta}_k$ captures its \emph{convergence rate}. Guided by a folklore in the ML community, we assume all convergence rates $\tilde{\beta}_k \in [0,1]$ and $\alpha_k > 0$. Our assumption here echos the ``scaling law'' in recent studies of large language models \citep{kaplan2020scaling}, where it is observed that the test loss scales as a power-law with model size and data size.
 
%decide their generation strategies and publish content online, a model-based creator can train a GenAI model (e.g., an LLM) from the collected published content data. An important assumption we make here is that the ability of this GenAI-based content generator is an increasing function of the total calibrated-body of content from all human creators, since it is widely known that the quality of GenAI depends on both the volume and quality of available human generated content on the market, and scales polynomially with respect to these two factors \fan{find citations} \hnote{Not sure if paper about the ``scaling law'' will be relevant.}. 
% Specifically, we model the competitiveness of GenAI on topic $k$ as $f_k(\sum_{j=1}^nx_{jk})=\alpha_k(\sum_{j=1}^nx_{jk})^{\tilde{\beta}_k}$, which is an increasing polynomial function of the total quality-calibrated body of human generated content $\sum_{j=1}^nx_{jk}$ on topic $k$. Here $\alpha_k\geq 0$ stands for the GenAI's data usage level, i.e., a large $\alpha$ means an extensive usage of GenAI in the content market. $\tilde{\beta}_k>0$ measures the efficiency of the training data usage of the GenAI model. In the case of $\alpha$ or $\tilde{\beta}$ is zero, the power of GenAI degenerates to a small constant and creators are competing in a GenAI-free environment. 

\begin{remark} A key abstraction in the above model is to condense a creator's (possibly complex) effort on any topic $k$ into a single value $x_{ik}$ capturing creator $i$'s  body of content (calibrated by qualities). 
At a first glance, it may appear overly simplistic; however, this modeling in fact characterizes many seemingly more general situations. For instance, suppose each topic $k$ has its own  topic embedding represented by a vector $\tt_k\in\RR^d$ whereas the generated content by creator-$i$ has some vector embedding $\mathbf{s}_{ik} \in\RR^d$ and $i$ may even have certain intrinsic competent level on topic $k$  represented by certain vector  $\w_{ik}\in \RR^d$. Suppose these factors together determine creator $i$'s body of contents $ x_{ik}= h(\mathbf{s}_{ik};\tt_k,\w_{ik})$. Let  $c_i(\{ \mathbf{s}_{ik} \}_k )$ denote  the total cost for $i$ to generate all the contents $\{ \mathbf{s}_{ik} \}_k $. We can easily reduce this more complex modeling to our cleaner abstraction by (1) letting $\mathbf{s}_{ik} = h^{-1} (x_{ik};\tt_k,\w_{ik})  $ be the   inverse of function $h$ w.r.t to    $\mathbf{s}_{ik}$; (2) viewing $ x_{ik}$ as $i$'s creation on topic $k$ instead of $\mathbf{s}_{ik}$; and (3) viewing $c_i(\{ h^{-1}(x_{ik}) \}_k ) = \tilde{c}_i (\x_{i}  ) $  as the new cost for creation $\x_{i}$.   
\end{remark}

\subsection{Context: Exclusive and Inclusive Competitions}\label{subsec:exclusive_inclusive}
Human creators and GenAI compete for user attention on each topic $k$. To be most general, we assume the total user attention/traffic on each topic $k$ is governed by the function \begin{equation}\label{eq:total-user}
\text{user traffic at topic $k$:} \qquad   \mu_k   \cdot  (s_k)^{\tilde{\gamma}_k}, 
\end{equation}
which depends on the \emph{trendiness} of topic $k$, described by a scalar $\mu_k(>0)$, as well as the total body of content $s_{k}$
% $s_k = \sum_{j=1}^nx_{jk}$ 
under topic $k$ with a growth rate $\tilde{\gamma}_k \in [0,1]$. The rate $\tilde{\gamma}_k $ is introduced to capture the fact that more user traffic will be attracted as the total volume of content increases,
% the total user traffic will increase with the total volume of content, 
but it will gradually saturate  as the volume becomes extremely large \citep{butler2014attraction,tafesse2023content}. In the content creation competition, this total user attention of \eqref{eq:total-user} will be split between human and GenAI creators. Next, we consider two different situations of the competition, which are motivated by the different stages of GenAI technology adoption in the market. 
%associated with topic $k$ is given by $\mu_k(\sum_{j=1}^nx_{jk})^{\tilde{\gamma}_k}$, which represents the maximum possible traffic or attention that can be exploited from this topic. It has a minimum value $\mu_k>0$ and is increasing with respect to the calibrated body $\sum_{j=1}^nx_{jk}$ with diminishing marginal return ($0\leq\tilde{\gamma}_k<1$), which models the effect that a larger pool of higher quality content usually contributes to a higher total user engagement \fan{cite}. In the special case where $\tilde{\gamma}_k=0$, creators are simply competing for a fixed amount of traffic.

\vspace{2mm}
\noindent
\textbf{Exclusive human-vs-GenAI competition.} In this case, we assume GenAI is a standalone creator who competes with the $n$ human creators. This models the situation at the early stage of GenAI adoption in a new market, where only few pioneering people/companies have the capability/resource to use the technology whereas majority of the creators are still counting on traditional approaches for content creation. Following the standard Tullock competition model, we assume human creator $i$ will attract $\frac{x_{ik} }{\alpha_k(s_k)^{\tilde{\beta}_k}+s_k}$ fraction of the total user traffic, and hence derive the following utility
\begin{align}\notag
    u_i(\x_i,\x_{-i})&=\sum_{k=1}^K \frac{x_{ik}\cdot \mu_k \cdot  (s_k)^{\tilde{\gamma}_k}}{\alpha_k(s_k)^{\tilde{\beta}_k}+s_k} - c_i(\x_i) \\ \label{eq:u}
    &=\sum_{k=1}^K \frac{x_{ik}\cdot \mu_k}{\alpha_k(s_k)^{\beta_k}+(s_k)^{\gamma_k}} - c_i(\x_i), 
\end{align} 
 where $\beta_k=\tilde{\beta}_k-\tilde{\gamma}_k \in [-1, 1]$ %\hf{Fan, is your equilibrium existence still correct for $\beta \in [-1, 1]$? I know it is corret for $\beta \in [0, 1]$. } \fan{We cannot derive monotonicity when $\beta<0$. So we will need more justification here.}\hf{That's okay, we will just assume it then.}
 and $\gamma_k=1-\tilde{\gamma}_k \in [0, 1]$ are more convenient notations and will be used henceforth. Let $\bm{\mu} = (\mu_1, \cdots, \mu_K)$, and we denote the game above as \exgame{}$(\bm{\alpha}, \bm{\beta},\bm{\gamma},\bm{\mu},\{ c_i \}_{i=1}^n)$.\footnote{A mild technical assumption we make about the game is that $u_i(\x_i,\x_{-i})<0$ when any   $x_{ik}\rightarrow +\infty$. This means the growth of cost of making infinite volume of content always outweighs the user traffic growth, and is needed for technical reasons.} Following the convention in the game theory literature, we study the pure Nash equilibrium (PNE) of this game \citep{nash1950equilibrium}, as defined below. 

%\hf{For definitions, we usually want to spell out the full name of any concept sincey you want to make is self-standing (e.g., when people just want to skim through the paper by looking at the few definitions and results, they should be self-constained. That's why I have kept full name such as pure Nash equilibrium in all definition, theorem statements since that can make it self-contained.)}

 \begin{definition}\label{def:PNE} A profile of human creator strategies $\{ \x^*_i \}_{i=1}^n$ forms a pure Nash equilibrium (PNE), if for every creator $i$, $\x^*_i$ is a best response strategy; formally, 
 \begin{equation}\label{eq:pne-def}
      u_i(\x^*_i,\x^*_{-i}) \geq   u_i(\x_i,\x^*_{-i}) \, \,  \text{ for every } \x_i \in \RR_{\geq 0}^K.
 \end{equation}
 \end{definition} 
As widely known, the PNE does not need to always exist, though it is often viewed as a good prediction about players' behaviors whenever it exists and is unique \citep{debreu1952social,fan1952fixed,glicksberg1952further}. Thus, a significant portion of the contest analysis literature focuses on studying the existence and uniqueness of PNE. 
Our analysis in this paper focuses on the behavior of human players. GenAI in our model is 
non-strategic and derives the body of content under each topic based on the production of 
all human creators as defined in \eqref{eq:AIcontent}. Thus it is not a strategic player. We believe it is an interesting future direction to study the incentive of GenAI creators and how that affects the competition.  

\begin{remark}  %A few remarks are worth mentioning. First, 
The special case $K=1$, $\alpha_k = 0$ or $\beta_k = 0$ of our model corresponds to the classic Tullock contest \citep{tullock2001efficient}. However, this model cannot serve our purpose of studying human-vs-GenAI competition, and thus we propose the strictly more general model above. To our knowledge, this model is novel and has not be examined in previous literature, despite extensive economics research  in Tullock contest in the past 40 years. This is partially due to the fact that in standard contest environments, there is seldom a player whose competitiveness depends on the accumulated competitiveness of all other players, which we believe is a novel and unique feature of content creation by Generative AI. Perhaps surprisingly, despite the widely known challenge of analyzing the equilibrium of the classic Tullock contest, later we are able to show the existence and uniqueness of the Nash equilibrium for our generalized model in natural parameter regimes.

\end{remark} 

\noindent
\textbf{Inclusive human-vs-GenAI competition. } As GenAI technology becomes more mature and universally accessible at a later stage,  every human creator can opt to use it for content creation. This motivates our analysis of another competition context where, besides creating their own body of content $\x_i$,  every creator $i$ now has an \emph{additional} option to use GenAI for content creation. We assume the cost of GenAI content is significantly lower than that of genuine content and thus simply normalize the former
% cost of using GenAI 
to $0$. 
 
Formally, the inclusive human-vs-GenAI competition augments each creator's action space to $\Y  = \RR_{\geq 0}^K \cup \{ \bot \} $, where $\bot$ denotes the action of using GenAI for content creation and its cost is set to $0$ (when compared to $c_i(\x_i)$). To distinguish creators' strategies in this different competition context, we use $\y_i \in \Y$ to denote each creator $i$'s strategy. Given strategy profile $\{ \y_i \}_{i=1}^n$, the utility of creator $i$ is
 \begin{align}\label{eq:u-incl}
 u_i(\y_i,\y_{-i}) = 
    \begin{cases}
\sum_{k=1}^K \frac{x_{ik}\cdot \mu_k}{ n^\bot  \cdot \alpha_k(s_k)^{\beta_k}+(s_k)^{\gamma_k}} - c_i(\x_i), \, \text{ if } \y_i = \x_i \in \RR_{\geq 0}^K\\
\sum_{k=1}^K \frac{\alpha_k(s_k)^{\beta_k}  \cdot \mu_k}{ n^\bot \cdot \alpha_k(s_k)^{\beta_k}+(s_k)^{\gamma_k}}  , \qquad \qquad  \text{ if } \y_i = \bot  
\end{cases} 
\end{align}
where $ n^\bot  = |\{i: \y_i = \bot \} |$ is the number of GenAI creators and $s_k$ is the total body of content created by (only) \emph{humans} -- that is, content created by GenAIs do not contribute to $s_k$.\footnote{Recent studies show that while a small amount of synthetic data could help improve the GenAI model, too much synthetic data (e.g., more than half) will lead to model collapse \citep{bertrand2023stability}. This is why we assume GenAI's capability only depends on the total body of human created content.}
% \footnote{Recent studies about whether GenAI models can be improved using synthetic data generated by themselves. While a small amount of synthetic data could help, studies have shown that too much synthetic data (e.g., more than half) will lead to model collapse \cite{bertrand2023stability}. This is why we assume that GenAI's capability only depends on the total body of human created content.}
% This is not a surprise since eventually using  the data generated by a model to train the model itself can hardly create knowledge too far away from the originally real input data. }  
The PNE of this new game is defined similarly as Definition \ref{def:PNE}, by revising \eqref{eq:pne-def} to allow $\bot$ as an additional action. We denote this game as \ingame{}$(\bm{\alpha}, \bm{\beta},\bm{\gamma},\bm{\mu},\{ c_i \}_{i=1}^n)$.   

\subsection{The Case of Separable Costs and $1$-D Competition}\label{sec:model:1-D}
To analyze   asymptotic  properties of the equilibrium, a useful  structural assumptions about the creators' cost functions is that the cost  is separable across different topics, i.e., $c_i(\x) = \sum_{k=1}^Kc_{ik}(x_{ik})$ for some convex $c_{ik}: \RR \to \RR$ function. Since a creator's utility from user traffic is also separable (see \eqref{eq:u} and \eqref{eq:u-incl}), separable costs  effectively ``disentangle'' the competition at different topics. Hence, our analysis can simply focus on the competition along each single topic, leading us to study the following $1$-dimensional ($1$-D) competition where each creator $i$'s action is simplified   to a scalar $x_i \in [0, \infty) $. This can be alternatively viewed as a special case of our general model with $K = 1$. We remark that  studying competition with $1$-dimensional effort value is not as restrictive  as one might first think  --- in fact, most previous studies of Tullock contests, including the seminal work by \citet{tullock2001efficient}, have $1$-dimensional efforts.     

\vspace{2mm}

\noindent
\textbf{The 1-D Competition. } In order to analyze how the content creation capability of human affects their strategies   and GenAI's level of dominance at equilibrium, for such 1-D competition, we consider cost function with the form
$$c_i(x_i) = c_i \cdot  (x_i)^\rho, $$ 
where parameter $c_i(>0)$ captures creator $i$'s capability of creating content whereas $\rho(\geq 1)$ is a common parameter to all players. This cost function form has been widely adopted in previous literature for modeling creator economy \citep{jagadeesan2023supply,hu2023incentivizing}.    Let $\bm{c} = (c_1, \cdots, c_n)$. Without loss of generality, we assume $c_1 \leq c_2 \cdots \leq c_n$; that is, creators are indexed from the most to the least efficient.  Utilities and equilibria are inherited from our definitions in Section \ref{subsec:exclusive_inclusive},
% for exclusive and inclusive competitions, 
by simply setting $K = 1$.  We denote this 1-D exclusive competition as \basegame{}$(\alpha, \beta,\gamma,\mu,\rho,\{ c_i \}_{i=1}^n)$.

\section{Exclusive Human-vs-GenAI Competitions}
To study the exclusive competition game \exgame{}, %$(\bm{\alpha}, \bm{\beta},\bm{\gamma},\bm{\mu},\{ c_i \}_{i=1}^n)$, 
the most fundamental question is, perhaps, whether this competition among human content creators will ever reach a certain stable outcome and, if so, which outcome. %Given the wide adoption of Nash equilibrium, 
We answer this question by studying the pure Nash equilibrium (PNE) of the game, as described in Definition \ref{def:PNE}. 
%achieve a stable outcome in which they are satisfied with their deterministic strategies. Such a state is called pure Nash equilibrium (PNE), arguably the most natural solution concept in game theory. A PNE  is a joint strategy profile $(\x_1,\cdots,\x_n)$ such that no creator $i$ can increase her utility by unilaterally deviating from $\x_i$. 

Our first main result  establishes that, under mild assumptions,  \exgame{} always admits a unique PNE.

\begin{theorem}\label{thm:monotone_G}
Consider any \exgame{}$(\bm{\alpha}, \bm{\beta},\bm{\gamma},\bm{\mu},\{c_i\}_{i=1}^n)$. If  $\bm{\beta} \in [0, 1]^K$, then the game  is a strictly monotone game  hence admits a unique pure Nash Equilibrium. %if $c_i(\cdot)$ is strictly convex, or $c_i(\cdot)$ is convex and $\sum_{k=1}^K \alpha_k>0$.
       % \item $\lim_{x_{ik}\rightarrow+\infty}u_i(\x_i,\x_{-i})<0, \forall i\in[n],k\in[K]$. \hf{hmm...what does this condition mean? Does not seem natural...}\fan{Previously we only need to assume the cost $c_i(x)$ goes to infinity. However, since now we assume $\mu$ is also increasing in $s$, $c_i(x)\rightarrow+\infty$ is no longer sufficient to discourage players from playing infinite $x_i$. What we need is that the speed of $c_i(x)\rightarrow+\infty$ is faster than the speed of $\mu\rightarrow+\infty$, which gives this condition... } \hf{I see, sounds good. PS: I suggest making these assumption outside your theorem statement, but in your model, then justify your assumption. We should make theorem statement very simple (which appears more elegant). I think a much more powerful theorem will just be a single sentence. 
       
       % \textbf{Under assumption A and B, any SOME-NAME game   admits a unique Pure Nash Equilibrium.} 
       
       % To prove this, you can try to set up the assumption A, B (and possibly other  conditions like $\bm{\gamma} \in [0,1]^K$, which we should assume throughout the paper since $\beta_i > 1$ does not seem natural?)  in your modeling section, while not here in the technical section  }
\end{theorem}

Note that the primary challenge in proving Theorem \ref{thm:monotone_G} is to show \exgame{} is a strictly monotone game, whereas the existence and uniqueness of PNE in such games is a classic result of \citet{rosen1965existence}. It is known that standard Tullock contest is monotone \citep{even2009convergence}, and we show that the extended version of our proposed \exgame{} preserves the monotonicity property. Recall from \eqref{eq:u} that $\beta_k  =\tilde{\beta}_k-\tilde{\gamma}_k \in [-1, 1] $, which is the difference between GenAI's convergence rate and the growth rate of the total user traffic on topic $k$ resulted from the volume of content under this topic. Thus the  assumption of $\bm{\beta} \in [0, 1]^K$ in  Theorem \ref{thm:monotone_G} means that the GenAI algorithm's convergence rate needs to be larger than user growth rate for each topic $k$. %\hnote{assume $\tilde\beta=0$, i.e., GenAI's competitiveness is constant, increased human effort might give them increased utility, and therefore continue encouraging them to increase their effort?} \hf{may need to justify why this is reasonable. }

Theorem \ref{thm:monotone_G} is interesting from multiple perspectives. First, it strictly generalizes previous equilibrium existence results in classic Tullock contest \citep{perez1992general,cornes2005asymmetric}, which corresponds to the special case with $K = 1, \bm{\alpha} = \bm{0} $ and $ \bm{\gamma} = \bm{1}$. To the best of our knowledge, both the model we develop and the equilibrium uniqueness result in Theorem \ref{thm:monotone_G} are  new. Second, the fact that \exgame{} is a strictly monotone game is  significant  because it is well-known that the PNE of monotone games can be found efficiently. In fact, many natural multi-agent online learning dynamics such as mirror descent \citep{bravo2018bandit}, accelerated optimistic gradient \citep{cai2023doubly}, and payoff-based learning \citep{tatarenko2020bandit} guarantee the last-iterate convergence to the unique PNE in strictly monotone games, even when players have mere bandit feedback information about  their utility functions. These results suggest that the PNE of \exgame{} is achievable if all creators use a reasonable update rule in their strategies. This observation not only makes this equilibrium a plausible prediction of real-world competition but also paves the way to our simulation-based studies in our experiments of Section \ref{sec:exp}, where we use multi-agent mirror descent with perfect gradient to numerically solve the PNE of \exgame{}.
Our proof of Theorem \ref{thm:monotone_G} starts from a classic characterization of strictly monotone games by   \citet{rosen1965existence}, known as the \emph{diagonal strict concavity} (DSC). We analyze the spectrum of the Hessian matrix of \exgame{} via quadratic decomposition, 
% where we apply quadratic decomposition techniques 
and show that DSC is satisfied when $\gamma_k,\beta_k\in[0,1],\alpha_k\geq0$, and $\alpha_k>0$ for some $k$. Full proof is given in Appendix \ref{app:proof_thm1}.

% \begin{corollary}
% The Multi-agent mirror descent algorithm \cite{bravo2018bandit} converges to the PNE of \exgame{} with probability $1$.
% \end{corollary}

%% file: theory.tex
\subsection{Equilibrium Properties of the 1-D Competition}\label{sec:sepa_cost}
We now take a closer look at the properties of the unique PNE. At the \emph{micro-level}, we are interested in how the human content creators' behaviors and utilities change with their content creation capabilities. At the \emph{macro-level}, we are interested in how the total body of content evolves as the total user traffic and total creator creation efficiency change. To theoretically study these questions \footnote{We  will   also revisit  these questions empirically in Section \ref{sec:exp}.}, we turn to the 1-D competition game \basegame{} described in Section \ref{sec:model:1-D}.
%where the creators' capabilities are naturally characterized by the order cost value  $0<c_1\leq \cdots \leq c_n$ and the exponent rate $\rho$. 
The following theorem illustrates multiple micro-level equilibrium properties at the unique PNE in the 1-D competition.

\begin{theorem}[Micro-level Equilibrium Properties]\label{lm:order_c}
The unique PNE $\x^*=(x^*_1,\cdots,x^*_n)$ of the game  \basegame{}$(\alpha, \beta,\gamma,\mu,\rho,\{ c_i \}_{i=1}^n)$ satisfies following properties:
\begin{enumerate}
    \item \textbf{Monotonicity of action and utility in creator capability}:     $x_1^*\geq \cdots\geq x_n^*$ and $u_1(\x^*)\geq \cdots\geq u_n(\x^*)$;
    \item \textbf{Monotonicity of utility in costs:}  if the $n$-th creator's cost increases from $c_n$ to $\tilde{c}_n$ while all other game parameters remain unchanged, then the new PNE $\tilde{\x}^*$ satisfies $u_n(\tilde{\x}^*)<u_n(\x^*)$;
    \item \textbf{Monotonicity of total creation in competition:} suppose a new player with cost $c_{n+1}$ joins the competition and $\x'=(x'_1, \cdots, x'_n, x'_{n+1})$ is the new PNE of \\\basegame{}$(\alpha, \beta,\gamma,\mu,\rho,\{ c_i \}_{i=1}^{n+1})$, it holds that 
    %\hf{What about $ \sum_{i=1}^{n+1} x'_i $ versus $ \sum_{i=1}^n x^*_i.$? }\fan{It's not clear, this side is much harder to prove and is not useful in our main theorems. }
    \begin{equation}\label{eq:414}
     \sum_{i=1}^n x'_i < \sum_{i=1}^n x^*_i. 
 \end{equation}
\end{enumerate}
\end{theorem}
Theorem \ref{lm:order_c} reveals three basic facts about the PNE of \basegame{}. First, a creator with higher creation cost tends to generate less content and receive lower utility at the PNE. Second, if one creator suffers from an increased cost, her utility decreases at the new equilibrium. The third property states that whenever a new creator joins and induces a new PNE, the volume of the original $n$ creators' content creation would decrease in response to the more competitive environment. All these properties are quite intuitive and insightful in a real-world competition environment and they will serve as technical tools in the proof of our main results. The proof of Theorem \ref{lm:order_c} is shown in Appendix \ref{app:proof_lm_order_c}.

Our next proposition predicts how an individual creator balances her gain from the traffic and the creation cost. 
\begin{proposition}[The utility--cost balance at equilibrium]\label{lm:cost_bound}
%Let $\x=(x_1,\cdots,x_n)$ be the PNE of \basegame{}$(\alpha, \beta, \gamma, \mu, \{c_i\}_{i=1}^n,\rho)$ \hf{I moved the $\rho$ to the last so that it is clear the last two components are what you just added whereas previous ones are the same as general game. Please make sure this is applied in general. } and $s=\sum_{i=1}^n x_i$ be the corresponding total calibrated body of content. Then for each creator $i$, her cost at the PNE satisfies the following inequality:
Let $\x^*=(x^*_1,\cdots,x^*_n)$ be the unique PNE of \basegame{}$(\alpha, \beta, \gamma, \mu,\rho, \{c_i\}_{i=1}^n)$ %\hf{I moved the $\rho$ to the last so that it is clear the last two components are what you just added whereas previous ones are the same as general game. Please make sure this is applied in general. Also, let's try to use $x^*$ to refer to equilibrium } 
and $s^*=\sum_{i=1}^n x^*_i$ be the  total   body of content. For each creator $i$, her cost at this PNE satisfies the following inequalities:
\begin{equation}\label{eq:cost_bound}
    \frac{1}{2\rho} \frac{x^*_i\cdot \mu}{(s^*)^{\gamma}+\alpha (s^*)^{\beta}}<c_i (x^*_i)^{\rho} < \frac{1}{\rho}  \frac{x^*_i\cdot \mu}{(s^*)^{\gamma}+\alpha (s^*)^{\beta}}.
\end{equation}
\end{proposition}

Note that in \eqref{eq:cost_bound}, the term $\frac{x^*_i\cdot \mu}{(s^*)^{\gamma}+\alpha (s^*)^{\beta}}$ is creator $i$'s gain from the traffic and $c_i (x^*_i)^{\rho}$ is her creation cost. 
Lemma \ref{lm:cost_bound} shows that at the PNE each creator will balance between their gain from user traffic and creation cost such that they only differ by a multiplicative factor between $1/(2\rho)$ and $1/\rho$.  Moreover, \eqref{eq:cost_bound} also suggests that when the marginal cost of creation increases (i.e., larger $\rho$), creators tend to choose a strategy $x_i$ that incurs a smaller cost compared to the gain, which is commonly observed in reality.

Our next main result is at the macro-level and reveals how the total calibrated body of content $s^*=\sum_{i=1}^n x^*_i$ created at the PNE $\x^*=(x_1^*,\cdots,x_n^*)$ is affected by  the game parameters $(\alpha, \beta, \gamma, \mu, \rho, \{c_i\}_{i=1}^n)$. The following \emph{Hadamard inverse}
%\dn{(This is the proper name for this)}\hf{Thanks Denis. Adopted!}
 of the cost vector $\c$ turns out to organically appear in our characterization:
%\begin{equation*}
   $ \c^{-1} = (c_1^{-1}, \cdots, c_n^{-1}) \in \RR_+^n.$ 
%\end{equation*}

Since $\c$ are the costs, $\c^{-1}$ can be naturally interpreted as the \emph{creation efficiency} of each creator, hence the larger the better.  For any vector $\c$, we use $\|\c\|_{ \rho}=\left(\sum_{i=1}^n c_i^{\rho}\right)^{1/\rho}$ to denote its   $L_{\rho}$-norm. It turns out that $s^*$  depends  on the total user traffic $\mu$ and a particular norm of the Hadamard inverse  of the cost vector $\c^{-1} $, as formalized below. 

% The order of $s^*$ turns out to depend  on the total demand $\mu$ and certain norm of the cost vector $\c=(c_1,\cdots,c_n)$. We use $\|\c\|_{ \rho}=\left(\sum_{i=1}^n c_i^{\rho}\right)^{1/\rho}$ to denote the   $L_{\rho}$-norm of the cost vector   $\c$, and our characterization can then be stated as follows. 

\begin{theorem}[Macro-level Equilibrium Properties]\label{thm:asym_s}
For any sufficiently large $\mu$ and $n$, the total calibrated body of human-created content $s^*$ at the PNE of \basegame{} satisfies
\begin{align}\label{eq:s_order} 
 \frac{C_{\rho}}{2\alpha+2}   <\frac{(s^*)^{\gamma + \rho -1}}{  \mu \cdot  \|\c^{-1}\|_{\frac{1}{ \rho - 1}}   }< C_{\rho}  
\end{align}
where  $C_{\rho}$ is a constant depending on $\rho$ (but not on $\beta$). % \hf{Where is $\beta$ in this theorem? They are not relevant? }\fan{they are not as relevant as $\alpha$ as when $\beta\rightarrow 0$, $s^{\gamma}+\alpha s^{\beta}\rightarrow s^{\gamma}$, when $\beta\rightarrow 1$, $s^{\gamma}+\alpha s^{\beta}\rightarrow (1+\alpha)s^{\gamma}$.}
% \begin{align}\notag
%     & s^*=\Omega\left(\left(\frac{\mu}{\alpha}\right)^{\frac{1}{\gamma+\rho-1}} \|\c^{-1}\|_{\frac{1}{\rho-1}}^{\frac{1}{\gamma+\rho-1}}\right) ~~\text{and} \\\label{eq:s_order} &s^*=\mathcal{O}\left(\mu^{\frac{1}{\gamma+\rho-1}} \|\c^{-1}\|_{\frac{1}{\rho-1}}^{\frac{1}{\gamma+\rho-1}}\right)
% \end{align}

% \begin{align}\label{eq:s_order} 
% C_{\gamma,\rho}\left(\frac{\mu}{2\alpha+2}\right)^{\frac{1}{\gamma+\rho-1}} <\frac{s^*}{\|\c^{-1}\|_{\frac{1}{\rho-1}}^{\frac{1}{\gamma+\rho-1}}}<C_{\gamma,\rho}\mu^{\frac{1}{\gamma+\rho-1}} 
% \end{align}

% is the unique solution of the following equation
% \begin{equation}\label{eq:s_beta}
%     \frac{a}{a+s^{1-\beta}}+\sum_{i=1}^n \frac{1}{1+a\beta s^{\beta-1}+c_i(s+as^{\beta})^2}=1.
% \end{equation}
% When $s$ and $n$ are sufficiently large (it's the volume of data and the number of creators), $s^*(\beta)$ is a decreasing function of $\beta$. 
\end{theorem}

Note that both the lower and upper bounds in \eqref{eq:s_order} are constants, and the cost function $c_i x^\rho$ is convex, as $\rho -1 \geq 0$. In real-world online content market, the total amount of traffic $\mu$ and the number of creators $n$ are 
indeed formidably large, hence usually satisfy the requirements of  Theorem \ref{thm:asym_s}. \eqref{eq:s_order} quantifies how the   total body of created content $s^*$ are determined by two macro-level  factors: 1) the overall trendiness of the topic $\mu$; and 2) the total creation efficiency of all human creators   $\|\c^{-1}\|_{\frac{1}{\rho-1}}$, which is precisely the $\ell_{\frac{1}{\rho-1}}$-norm of the production efficiency vector $\c^{-1}$. In particular, the order of their product has to be exactly $(s^*)^{\gamma + \rho -1}$, with lower and upper bound constants specified in \eqref{eq:s_order}. Hence the larger any factor is, the larger $\s^*$ is. \footnote{This is also why $\c^{-1}$ (not $\c$) organically arrives in the result.}  

% hence the bigger one factor is, the smaller another is. This tension implies that $s^*$  is determined by $\mu$ and total cost effectiveness $\|\c\|_{\frac{1}{1-\rho}}$ through the following form: }
% \begin{equation}\label{eq:s-star-dependence}
%     s^* = O  \bigg( \Big[  \mu \cdot  \|\c\|^{-1}_{\frac{1}{1-\rho}} \Big]^{\frac{1}{\gamma + \rho -1}} \bigg).
% \end{equation}
%where the omitted big-$O$ constant is within the upper and lower bound in Eq. \eqref{eq:s_order}.  
%\hf{We are using Eq. and Equation exchangeably. Please unify it using just one form throughout the paper. }

A few important insights can be learned from Theorem \ref{thm:asym_s}. First, note that GenAI's learning rate $\beta$ did not show up in \eqref{eq:s_order} whereas the GenAI's data efficiency parameter $\alpha$ does, though it only mildly effects the denominator of the lower bound term. We view this as encouraging message to human creators, as even though GenAI  may potentially outperform any individual human creator (i.e., volume of content created by GenAI can be significantly larger than any $x_i$), on the macro-level they will only affect the total body of created content by up to a constant factor. \emph{From the macro perspective, we view this observation as a  symbiosis rather than fundamental conflicts  between GenAI and human creators, answering the question raised in the Introduction section.}  However, we  remark that  this macro-level characterizations do  not rule out the possibility that, at the individual level, some creators may do significantly worse. % \hnote{But this only suggests the humans will not create less, collectively?}

Second, Theorem \ref{thm:asym_s} helps us understand how the total body of content evolves as either creator popularity or total platform user traffic change. The term $\|\c^{-1}\|_{\frac{1}{\rho-1}}$ (recall $\rho \geq 1$) increases when either (a) any individual creator's cost decreases or (b) more creators join the competition. These situations will always lead to more human-created content.  % and therefore can serve as a measurement for the total cost in the competition. In this sense, Eq. \eqref{eq:s_order} indicates an inverse relationship between the total cost and the total volume of content generation given a fixed content consumer demand $\mu$. 
Moreover, Theorem \ref{thm:asym_s} offers additional insights about how the total content creation scales with respect to the trendiness of the topic $\mu$ and human production efficiency characterized by $\|\c^{-1}\|_{\frac{1}{\rho-1}}$. For more insight, a useful special case to consider is when  every creator's cost $c_i$ is around some constant $c$, then the total creation efficiency  $\|\c\|^{-1}_{\frac{1}{1-\rho}}$ is around  $c^{-1} n^{\rho - 1} $ and  \eqref{eq:s_order} can be simplified to $s^*=O\left(   (\frac{\mu}{c})^{ \omega }    n^{1 - \gamma \omega}\right) = O\left(   (\frac{\mu}{c n ^{\gamma}})^{ \omega }    n   \right) $, where $\omega =  \frac{1}{\rho + \gamma - 1}$. 
An interesting observation is that, as creation cost rate $\rho$ decreases, $\omega$ increases; the total body of content $s^*$ will increase when the overall trendiness $\mu$ outweighs certain competition level $c n^{\gamma}$ of the population, but will decrease when $\mu$ becomes smaller than  $c n^{\gamma}$. \emph{This shows an intriguing double-edged  effect of human creation   efficiency   on the total body of content} --- better efficiency   increases content creation when there is sufficient user demand but, somewhat surprisingly, it will reduce  content creation when the user demand is not sufficient. 
Moreover, when $\rho (\geq 1)$ decreases to approach $1$, $\omega \to 1/\gamma$ and $s^*$ will become almost linear in the total user traffic $\mu$.      
%It implies that the total volume of content creation increases polynomially in response to market demand $\mu$ and the number of creators $n$, with the growth rate being influenced by the marginal cost of human content creation, $\rho$. As $\rho$ decreases and approaches 1, the total content volume exhibits near-linear growth relative to $\mu$, representing an optimal scenario. Finally, these finds all suggest that even GenAI-based creators can potentially outperform any individual human creator (for certain $x_i$ values), they cannot dominate the collective human effort. 

The proof of Theorem \ref{thm:asym_s} is based on various characterizations of PNE and  properties under large $n,\mu$, e.g., no one at PNE constitutes a constant fraction of the total body of content. We defer formal discussions to Appendix \ref{app:proof_thm2}.

\section{Inclusive Human-vs-GenAI Competitions}

In this section, we turn to the study of inclusive competition \ingame{}, in which GenAI will now serve as an accessible action to each creator rather than a special agent with exclusive power. We similarly start by studying the PNE of the competition. Unfortunately, unlike the case of exclusive competition, the PNE of \ingame{} may not exist in general if the cost functions are non-separable, as shown below.
% , as we shown in the following theorem. 
\begin{theorem}\label{thm:non_exist_PNE}
    The pure Nash equilibrium of \ingame{} needs not exist when $K>1$, even when $\{c_i\}_{i=1}^n$ are strongly convex.
\end{theorem}
To prove this theorem, we explicitly construct an instance of \ingame{} in Appendix \ref{app:non_exist_PNE}, 
 and show that this game does not admit any PNE.  Our constructed instance also illustrates an interesting tension of competition in practical scenarios. The instance has two players, $A$ and $B$, who are constantly switching between adopting GenAI technology and using traditional content creation method. When creator $A$ switchs to GenAI, it prompts the human creator $B$ to reduce her effort, in response to the intensified competition. However, as $B$ reduces her effort, the available  training data for $A$'s  GenAI deteriorate, leading $A$ to eventually forego GenAI usage. This shift creates an opportunity for $B$ to adopt GenAI by herself, adversely affecting $A$'s competitiveness. Hence $A$ alters her strategy, which then causes $B$ to abandon GenAI and revert to traditional content creation. This cycle results in an endless loop of best responses, thereby negating the possibility of any possible PNE. By viewing the $A$ and $B$ above as representative groups of human creators, we believe that the tension revealed in this constructed instance also presents the real-world scenario. 

\subsection{Equilibrium Properties of the 1-D competition}
Given the   non-existence of PNE in general \ingame{} game, we naturally turn our attention to the separable-cost case hence 1-D competition, with the hope of restoring existence of PNE in this special case. 
%We denote such a competitive environment \ingamed{}$(\alpha,\beta,\gamma,\mu,\{c_i\}_{i=1}^n, \rho)$. 
Fortunately,  this indeed turns out to be possible, as shown in the following theorem.   Recall that we assume the cost function has the form $c_i x^\rho$, and the creators are sorted so that $0<c_1\leq \cdots\leq c_n$.  Then we can denote the resultant game as \ingamed{}$(\alpha,\beta,\gamma,\mu, \rho,\{c_i\}_{i=1}^n)$.  
\begin{theorem}\label{thm:PNE_extend}
    Suppose $\tilde{\beta}+\tilde{\gamma}\geq 1$, then \ingamed{} always admits a pure Nash equilibrium (PNE) with form $\y^*=(x_1,\cdots,x_{n-m}, \bot,\cdots,\bot)$. That is, the creators using GenAI action $\bot$  at  PNE are those with top-$m$ highest costs (i.e., the least efficient $m$ creators).  
    
    Moreover,  for large enough $n,\mu$, the fraction of GenAI creators at this PNE has the following asymptotic lower bound: 
    \begin{equation}\label{eq:86}
        \frac{m}{n}>1-C\cdot\frac{\mu^\frac{\gamma-\beta}{\gamma+\rho-1}}{\alpha n^{1-\frac{(\gamma-\beta)(\rho-1)}{\gamma+\rho-1}}},
    \end{equation}
    where $C$ is a constant depending on $(\beta,\gamma,\rho,c_1)$.
\end{theorem}

Theorem \ref{thm:PNE_extend} delivers several important messages. First, it restores the existence of pure Nash equilibria in \ingamed{} for the 1-D competition case and postulates that those with the larger costs (i.e., less efficiency in human content creation) would switch to GenAI.  Although we are not able to  characterize all the PNEs of \ingamed{} rigorously, we will show in our experiments that multiple PNEs can exist but they all seem to share a similar property as   we identified theoretically in Theorem \ref{thm:PNE_extend}: \emph{creators with higher costs are more likely to resort to GenAI technology at equilibria.} 

Second, Theorem \ref{thm:PNE_extend} offers predictions about the conditions under which GenAI players may dominate the content market. It is easy to verify that the RHS of \eqref{eq:86} is increasing w.r.t. $\alpha$, $\beta$ and $n$. Therefore, when GenAI becomes increasingly powerful or the number of creators is very large, more creators will switch from the traditional human content creation to simply adopting the GenAI content creation. In addition, Ineq. \eqref{eq:86}  shows how the proportion of GenAI creators is affected by the ratio between the size $(n)$ and capacity $(\mu)$ of the content market: when $n\rightarrow +\infty$ or the growth of $n^{1-\frac{(\gamma-\beta)(\rho-1)}{\gamma+\rho-1}}$ dominates $\mu^\frac{\gamma-\beta}{\gamma+\rho-1}$, the RHS of Ineq. \eqref{eq:86} approaches $1$. This observation suggests that \emph{when the growth of total user traffic is diluted by an even more rapidly growing number of content creators, GenAI may become  the better choice for almost every creator. In this case, only the very few top creators with the best efficiency will still generate authentic content.}

%% file: experiment.tex
\section{Experiments}\label{sec:exp}

Since \exgame{} always has a unique PNE and multi-agent mirror descent \citep{bravo2018bandit} provably achieves such a PNE, we can empirically analyze its properties which is otherwise difficult to do theoretically. In our experiments, we instantiate concrete games and find their PNEs using simulations to observe how the rise of GenAI affects human content creators. Three types of games are considered: the exclusive competition \basegame{}, the inclusive competition \ingamed{}, and exclusive competition under non-separable cost \exgame{}.
\begin{figure}[t]

\includegraphics[width=0.49\columnwidth]{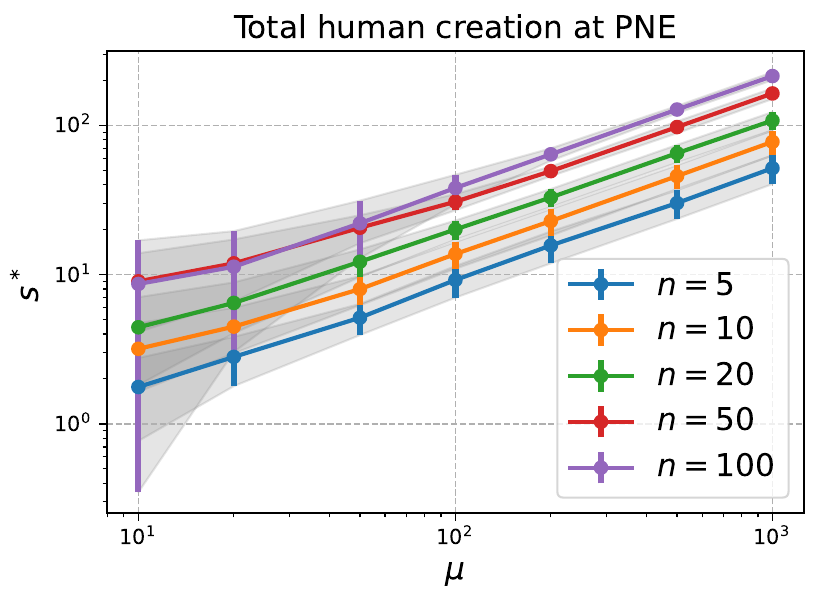}
\includegraphics[width=0.485\columnwidth]{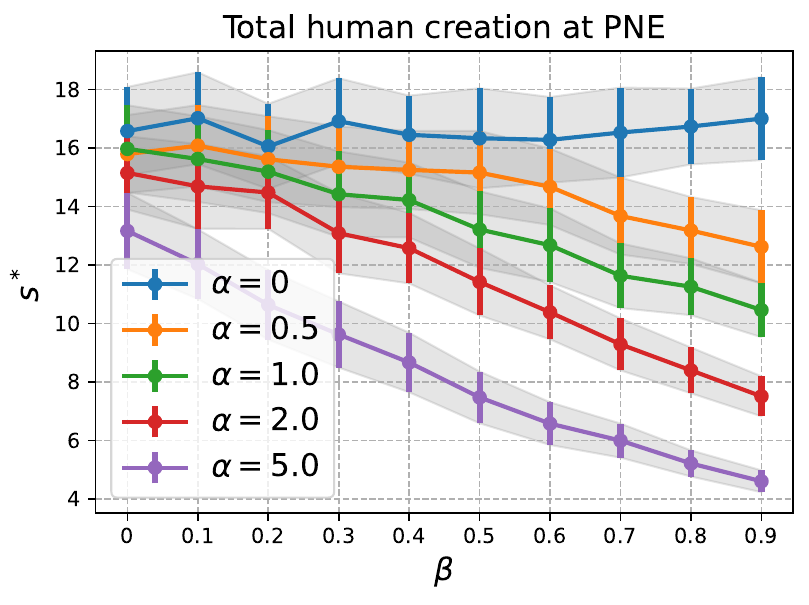}

\caption{The total body of content creation $s^*$ at the PNE as a function of $\mu,n$ (Left) and $\alpha,\beta$ (Right).}
\label{fig:beta_s}

\end{figure}

%\subsection{Simulation Environment}
\noindent\textbf{Simulation Environment.} The cost function for \ingamed{} and \basegame{} is set to $c_i(x)=c_ix^{\rho}$ (for non-separable cost function, we use $c_i(\x)=c_i\|\x\|_1^{\rho}$).  For \basegame{} and \ingamed{}, the default parameters are set to $n=10,\alpha=1.0,\beta=0.5,\gamma=0.9,\rho=1.5,\mu=100$, and $\{c_i\}_{i=1}^n$ are randomly sampled from uniform distribution $\U[1, 10]$. For \exgame{}, the default $K=10$ and $\{\alpha_k,\beta_k,\gamma_k,\rho_k\}^K_{k=1}$ are set to the same values as $\alpha,\beta,\gamma,\rho$. More results with heterogeneous parameters are presented in the appendix. The cost $\{c_i\}_{i=1}^n$ are randomly sampled from $\U[1, 10]$. In the subsequent experiments, when we investigate the sensitivity of the PNE on a certain parameter, we use the specified values to replace the default ones. Otherwise, we use default parameters to construct independent game instances and aggregate statistics from the resulting stochastic environments. The error bars in all results are obtained from 10 independent game instances. The details of the PNE solvers are given in Appendix \ref{app:pga_solver}. In the following, we study an array of interesting and important questions for understanding creators' collective behaviors under the influence of GenAI, and connect our theoretical results with the empirical findings.

\vspace{2mm}
\noindent
\textbf{Q1: How will the market size and GenAI's capability affect the volume of created content? }

\begin{figure}[t]
\includegraphics[width=0.49\columnwidth]{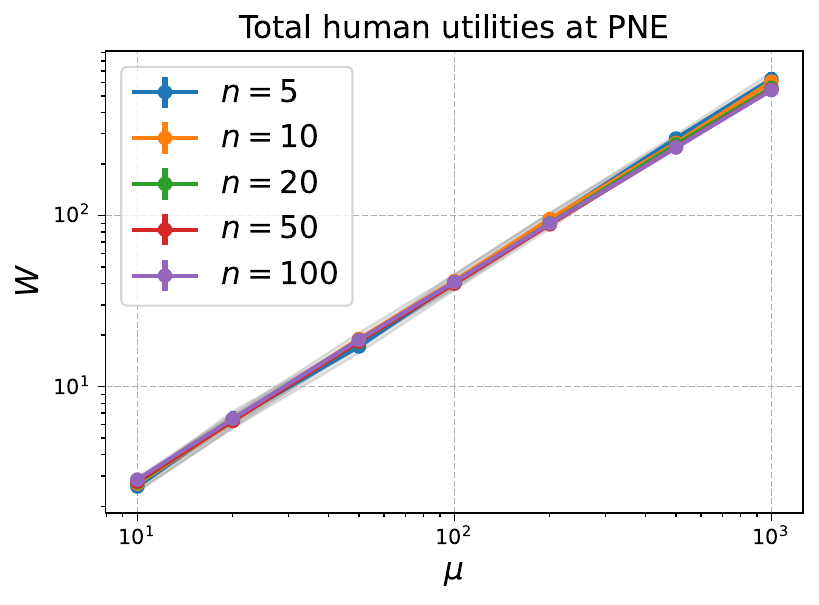}
\includegraphics[width=0.482\columnwidth]{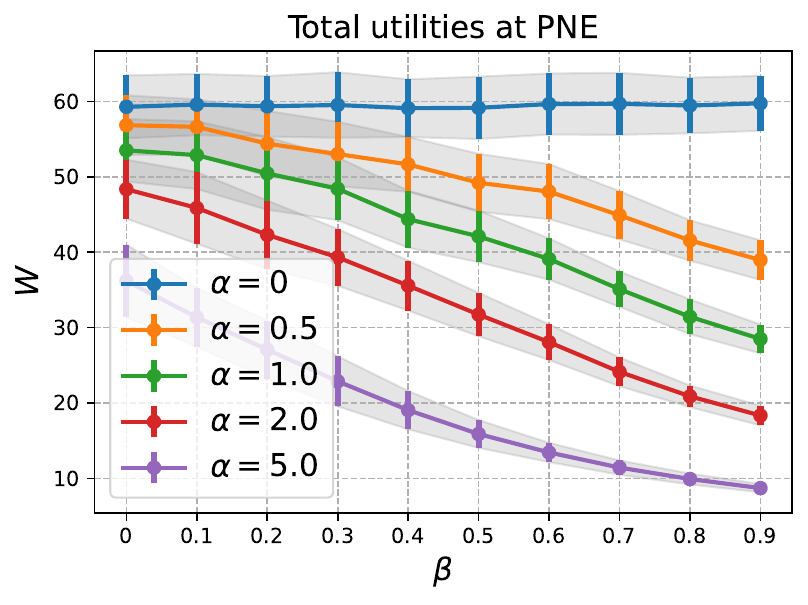}

\caption{The total human utilities $W=\sum_i u_i$ at the PNE as a function of $\mu,n$ (Left) and $\alpha,\beta$ (Right). Default $(\alpha,\beta)=(1.0, 0.5)$.}
\label{fig:beta_W}

\includegraphics[width=0.49\columnwidth]{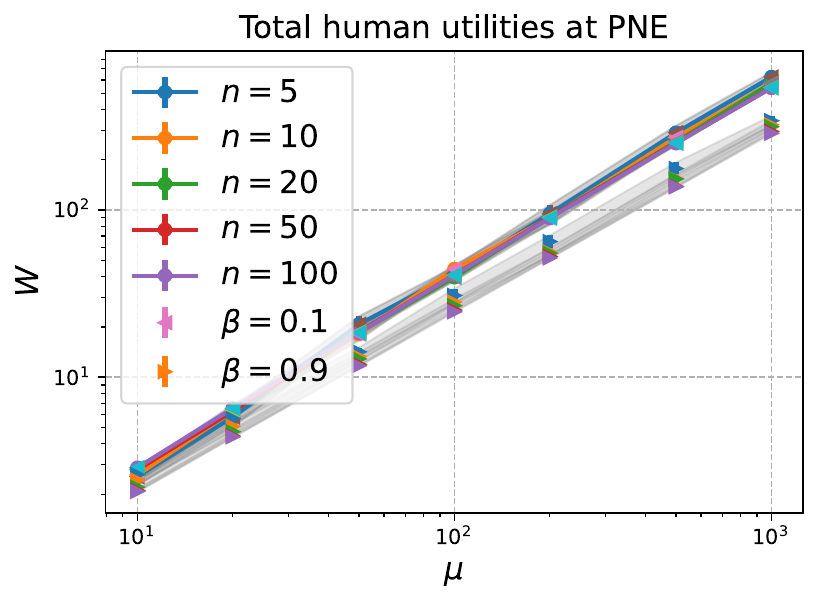}
\includegraphics[width=0.49\columnwidth]{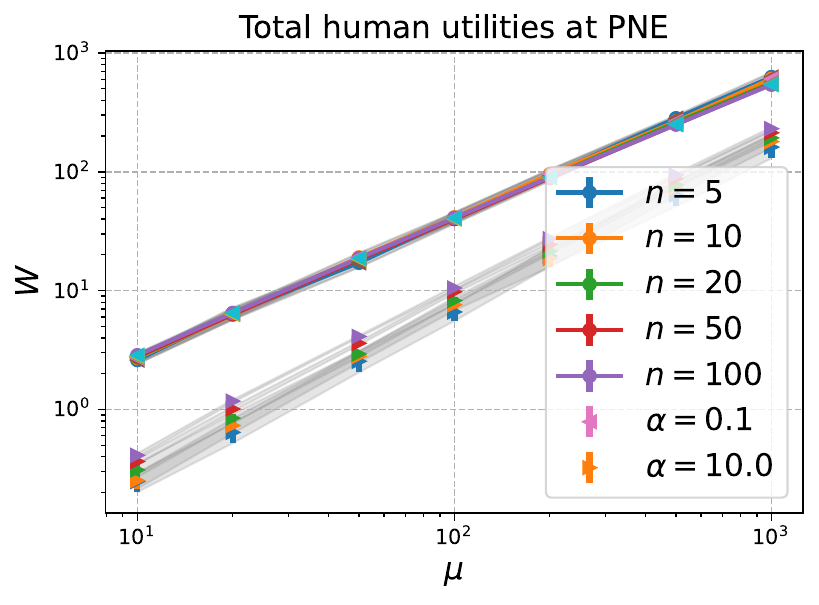}

\caption{The total human utilities $W=\sum_i u_i$ at the PNE as a function of $\mu,n$ (Left) and $\alpha,\beta$ (Right). Default $(\alpha,\beta)=(1.0, 0.5)$. The left panel shows results when $\beta \in \{0.1, 0.5, 0.9\}$ and the right panel shows results when when $\alpha \in \{0.1, 1.0, 10.0\}$. }
\label{fig:mu_W}
\end{figure}

Figure \ref{fig:beta_s} illustrates how the total human content generation at PNE changes w.r.t the market size $(\mu,n)$ and GenAI's ability $(\alpha,\beta)$. The log-log plot in the left panel indicates that the total content creation volume $s^*$ grows polynomially w.r.t $\mu$ under different $n$, when the engagement level of GenAI in the content market is fixed, which validates Theorem \ref{thm:asym_s}. However, the right panel shows that if the GenAI's ability increases with a larger $\alpha$ or $\beta$, the total effort from human players shrinks. 
%Interestingly, the total utilities (usually referred as the social welfare) share the same trend as $s^*$ in terms of $\alpha, \beta$ and $n$, which means a more extensive usage of GenAI not only adversely affects existing human creators' productivity but also their total welfare. 

If we define the social welfare as the total creator utilities $W=\sum_i u_i$, Figure \ref{fig:beta_W} and \ref{fig:mu_W} show how the social welfare (defined as the sum of creators' utilities) changes with respect to parameter $\mu,n,\alpha, \beta$. Similar to the result illustrated in Figure \ref{fig:beta_s}, the social welfare displays the same trend: it decreases as $\alpha,\beta$ increase and increases when $\mu$ increases. This suggests a more extensive usage of GenAI not only adversely affects existing human creator's productivity but also their utilities. 
Notably, the left panel of Figure \ref{fig:beta_W} shows the welfare does not rely on $n$ significantly when fixing other parameters, which means when the competition environment is set (including the strength of GenAI power/usage and the total user traffic), all human creators as a whole can neither benefit nor suffer from the magnitude of populations.

\noindent
\textbf{Q2: Contributors vs free-riders: which human creators will switch to the (free) GenAI creation?}

We simulate the PNEs of \ingamed{} to answer this question. First, we fix $n=100,\mu=1000$ and investigate how the ability of GenAI affects the number of creators who switch at the equilibrium. To solve the PNE guaranteed by Theorem \ref{thm:PNE_extend}, we start with an equilibrium where no one uses GenAI and then take advantage of its ordered cost structure by iteratively determining the set of GenAI creators. Specifically, we rank all creators in terms of their cost from high to low, and examine whether the human creator currently with the highest cost wants to switch to $\bot$. If so, we update the remaining human creators' strategies by solving a resulting \basegame{} instance and start the next round by examining the next human creator, until we find a state where no one wants to update her strategy. This PNE solver is summarized in Algorithm \ref{alg:pne1} in Appendix \ref{app:pga_solver}. From the results shown in Figure \ref{fig:beta_num_dropout_pne1}, at the PNE the proportion of GenAI creators grows as $\alpha,\beta,n$ increases and declines when $\mu$ grows, which aligns with the prediction in Theorem \ref{thm:PNE_extend}.

\begin{figure}[t]

\includegraphics[width=0.49\columnwidth]{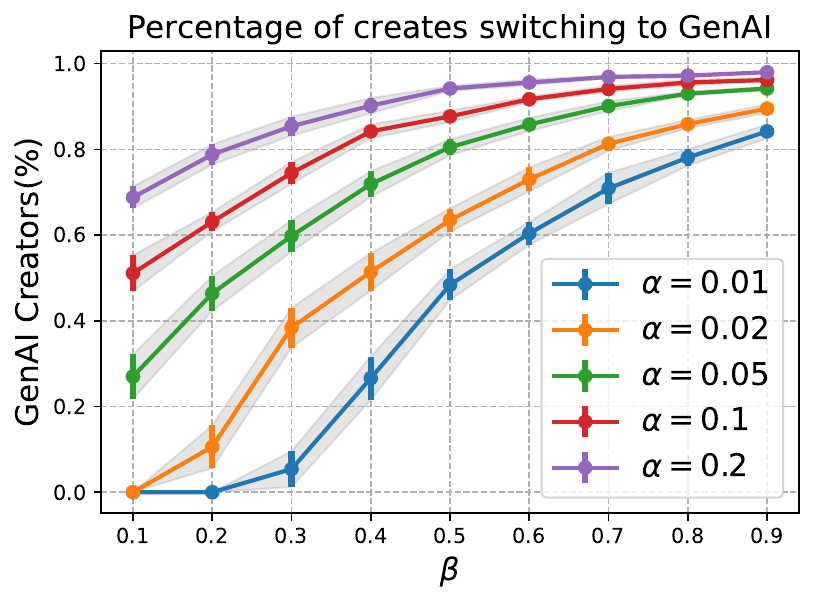}
\includegraphics[width=0.496\columnwidth]{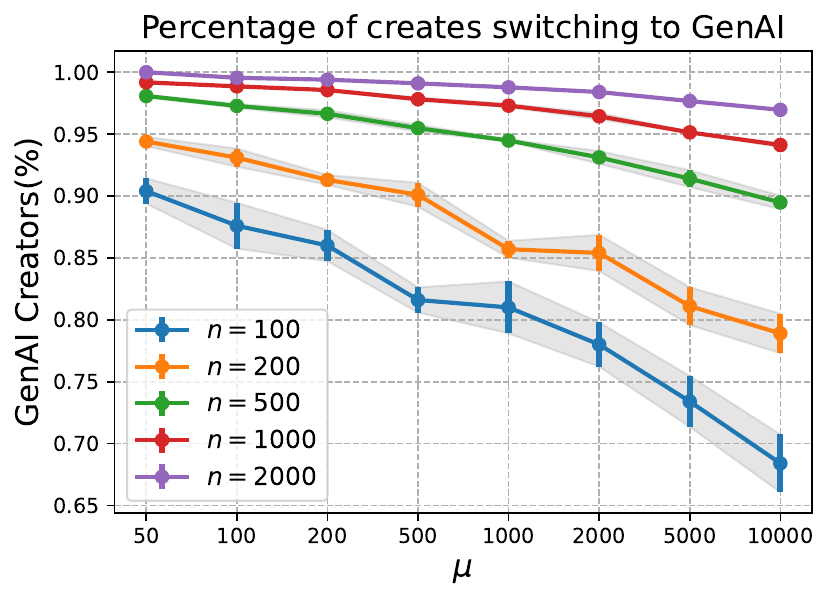}

\caption{The percentage of GenAI creators at PNE of \ingamed{} under different $\alpha,\beta$ or $\mu,n$.}
\label{fig:beta_num_dropout_pne1}

\end{figure}

Since the PNE of \ingamed{} is not unique, it is also interesting to examine whether other PNEs of \ingamed{} share similar characteristics as the one predicted by Theorem \ref{thm:PNE_extend}. To obtain an arbitrary PNE of \ingamed{}, we adapt the PNE solver in Algorithm \ref{alg:pne1}: at each iteration we pick a random creator to examine her type instead of always choosing the one with the highest cost, and apply an equilibrium checker to verify whether we arrive at a PNE. The detailed algorithm is summarized in Algorithm \ref{alg:pne_checker} and \ref{alg:pne2} in Appendix \ref{app:pga_solver}. We randomly generate $100$ \ingamed{} instances with their cost vectors $\c$ independently sampled from $\U[1, 10]$. For each instance, We grouped 100 creators into ten categories based on their costs and applied the revised algorithm to determine an arbitrary PNE, subsequently analyzing the distribution of GenAI creators across these groups. 
Our findings, as illustrated in Figure \ref{fig:beta_num_dropout_pne2}, confirm our anticipated trend that creators with higher costs are more inclined to adopt GenAI. This indicates a cost-dependent threshold influencing creators' choice of strategy.
%, with high-cost creators favoring GenAI, while those with lower costs prefer traditional methods. 
However, the choice becomes less predictable in the middle of cost range, 
%with the decision to use GenAI or stick to conventional strategies 
which varies according to specific conditions and the stochastic nature of reaching equilibrium.

\begin{figure}[t]
\includegraphics[width=0.49\columnwidth]{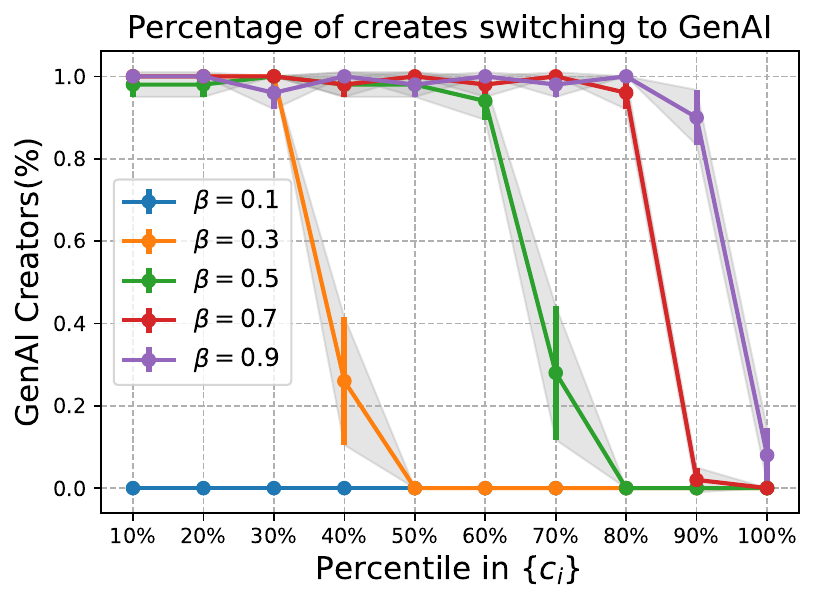}
\includegraphics[width=0.49\columnwidth]{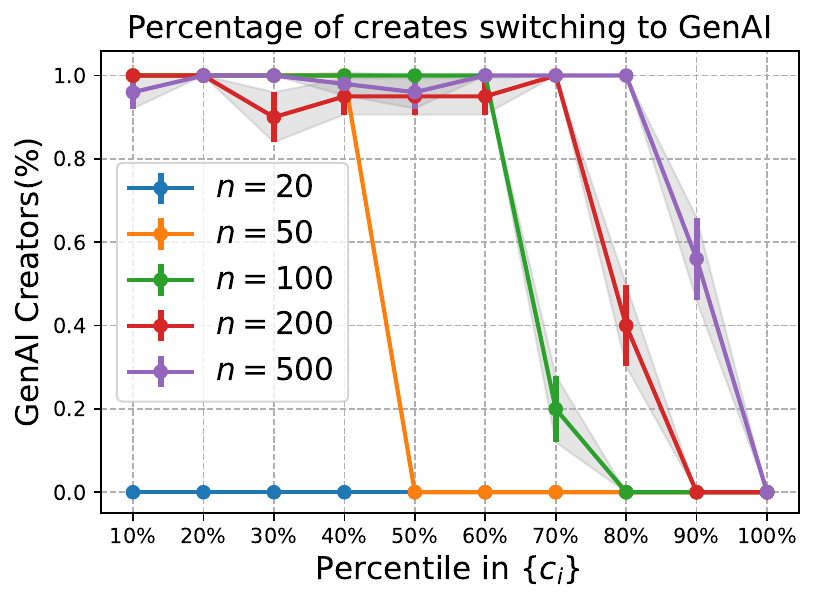}

\caption{How likely a creator with cost ranked at different percentiles tend to switch to GenAI at an arbitrary PNE of \ingamed{}.}
\label{fig:beta_num_dropout_pne2}

\end{figure}

\vspace{2mm}

\noindent
\textbf{Q3: Would the rise of GenAI-based creators discourage the productivity and utility of human creators?}

We compare the total content creation volume and the gained utility by both human and GenAI creators at the PNE of \ingamed{} obtained by Algorithm \ref{alg:pne1}. The left panel of Figure \ref{fig:beta_u_total_avg} illustrates a clear trend: with an increase in GenAI's capabilities (larger $\beta$), the total utility and content generation by human creators decrease, while those by GenAI creators see an uptick. Intriguingly, the average utility and content output per human creator actually improve and may even exceed those of GenAI creators as $\beta$ approaches $1$. This is because a more potent GenAI pushes higher-cost human creators out of the market, leaving behind a select group of efficient human creators who thrive due to reduced competition. This scenario suggests a likely future where the widespread adoption of GenAI substitutes less efficient content creators, capturing a significant share of viewer engagement. Nevertheless, a small, professional group of human creators will not only persist but flourish, producing more and higher-quality content. This scenario highlights the dual impact of GenAI: it eliminates the need for humans to perform monotonous tasks and fosters a more vibrant and rewarding ecosystem for the remaining, highly skilled creators.

\begin{figure}[t]
\includegraphics[width=0.49\columnwidth]{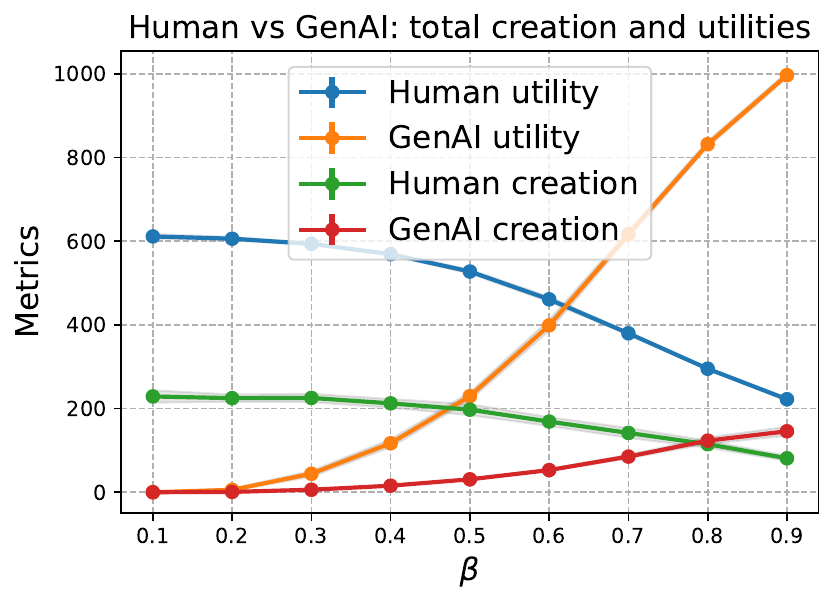}
\includegraphics[width=0.49\columnwidth]{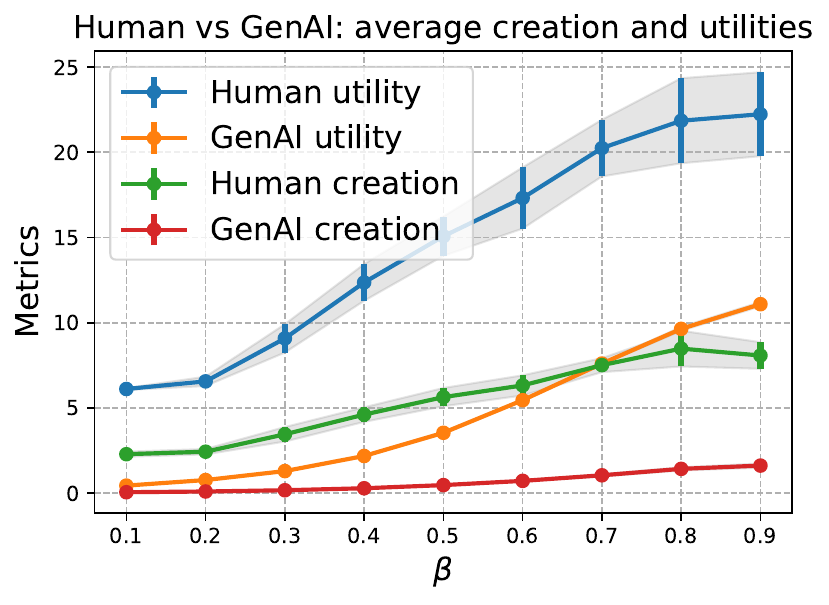}

\caption{Left/Right: Total/average utility and content creation of human creators vs GenAI creators.}
\label{fig:beta_u_total_avg}

\end{figure}

 \begin{figure}[t]
\includegraphics[width=0.49\columnwidth]{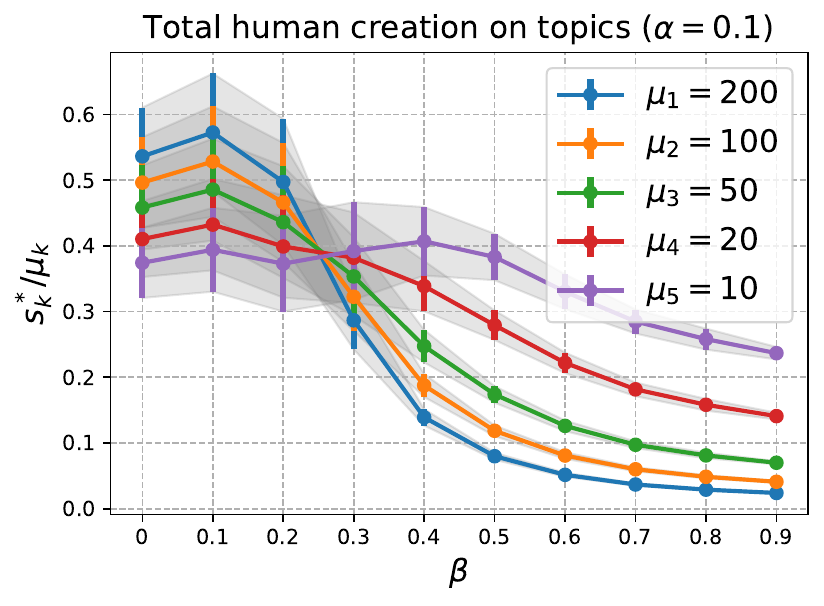}
\includegraphics[width=0.49\columnwidth]{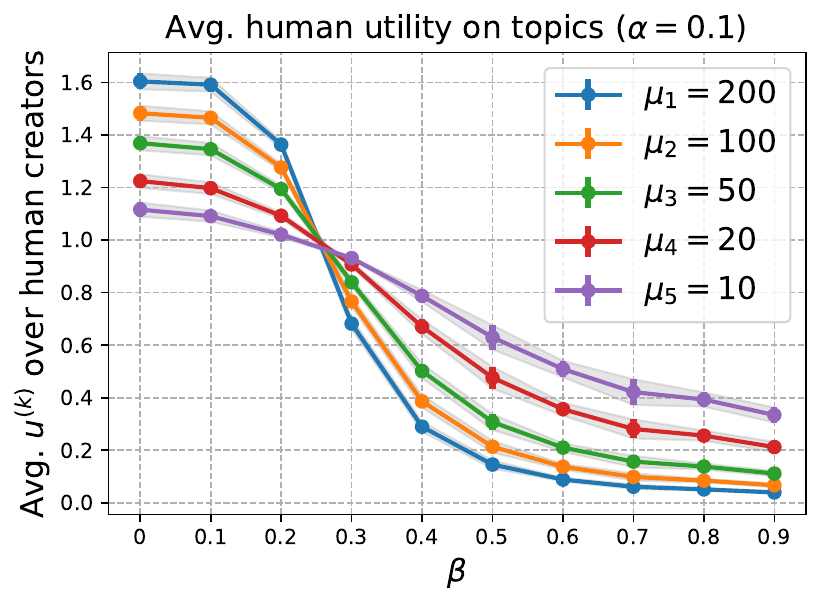}

\caption{$\alpha=0.1$. Left: the occupation ratio on each topic. Right:  average per-topic gain. Cost uniformly sampled from $\U[0.1,1.0]$.}
\label{fig:beta_s_non_3}

\end{figure}

\vspace{2mm} 
\noindent
\textbf{Q4: How would creators allocate efforts across different topics under non-separable cost functions?}

Now we investigate how GenAIs influence creators to strategically balance their effort on different domains. We use \exgame{} with $K=5$ and Algorithm \ref{alg:mmd} to simulate its PNE given $\mu=(200,100,50,20,10)$ with different $\alpha$ and $\beta$. Here we can view a topic with larger $\mu_k$ as a popular domain with higher total traffic, while topic with smaller $\mu_k$ represents a niche domain. To measure each human creator's inclination to different topics, we introduce the concept of ``occupation ratio'' to quantify human creators' engagement in a topic relative to its popularity, and we also analyze the creator's average utility gained from each topic.

Our findings illustrated in Figure \ref{fig:beta_s_non_3} reveal a notable shift in human creators' strategies as GenAI becomes more powerful. Initially, humans tend to focus more on popular domains, where the potential for traffic—and thus utility—is higher. However, as GenAI's strength increases, this preference shifts. Beyond a certain $\beta$ value, human creators pivot towards niche topics, resulting in greater relative gains. This strategic shift is attributed to GenAI's heightened competitiveness in popular domains (due to the increased volume of human created content), prompting humans to seek better opportunities in less contested areas.

\begin{figure}[t]
\includegraphics[width=0.49\columnwidth]{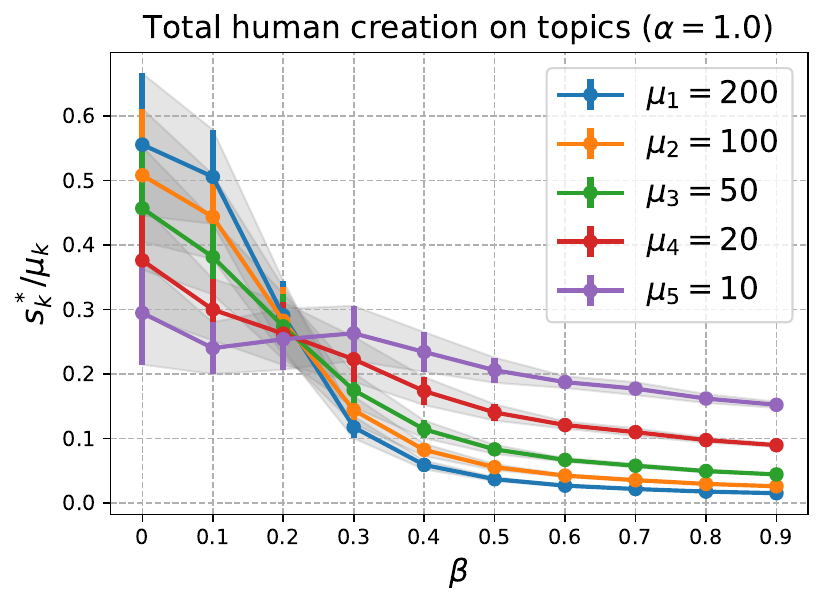}
\includegraphics[width=0.49\columnwidth]{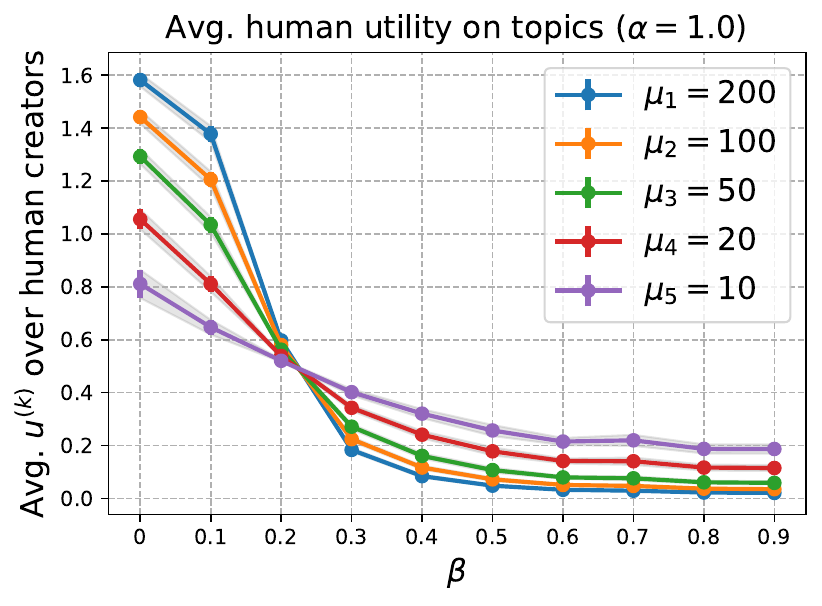}
\caption{$\alpha=1.0$. Left: the occupation ratio on each topic. Right: the average per-topic gain. Cost uniformly sampled from $\U[0.1,1.0]$.}
\label{fig:beta_s_non_4}

\includegraphics[width=0.49\columnwidth]{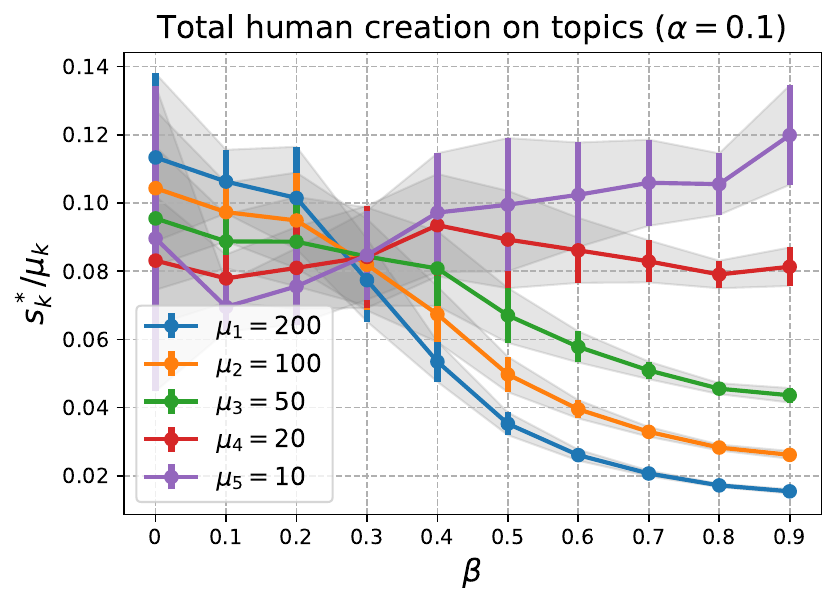}
\includegraphics[width=0.485\columnwidth]{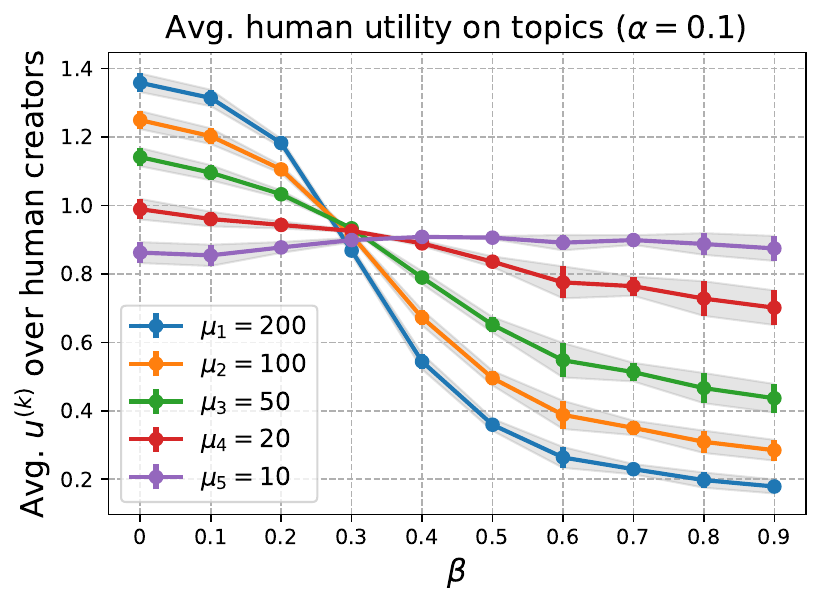}
\caption{$\alpha=0.1$. Left: the occupation ratio on each topic. Right: the average per-topic gain. Cost uniformly sampled from $\U[1.0,10.0]$.}
\label{fig:beta_s_non_2}
\end{figure}

\begin{figure}[h]
\includegraphics[width=0.49\columnwidth]{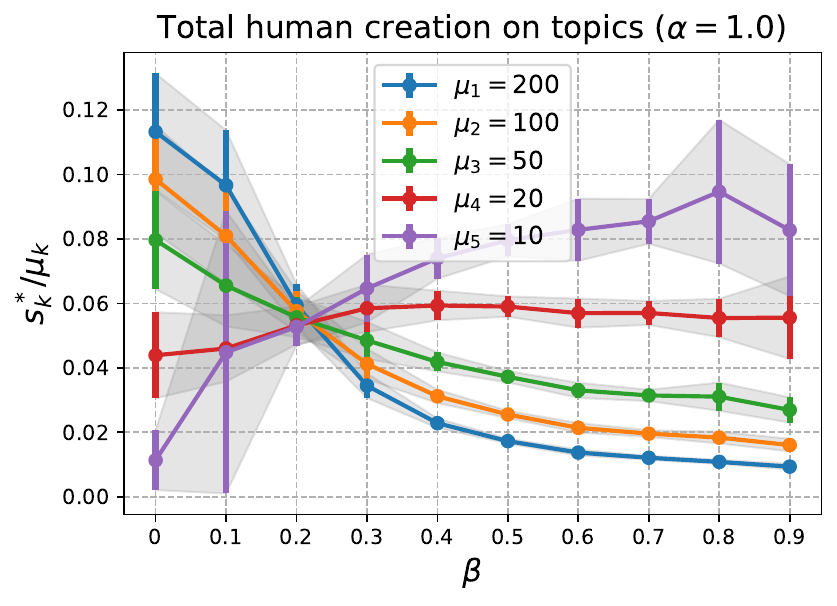}
\includegraphics[width=0.49\columnwidth]{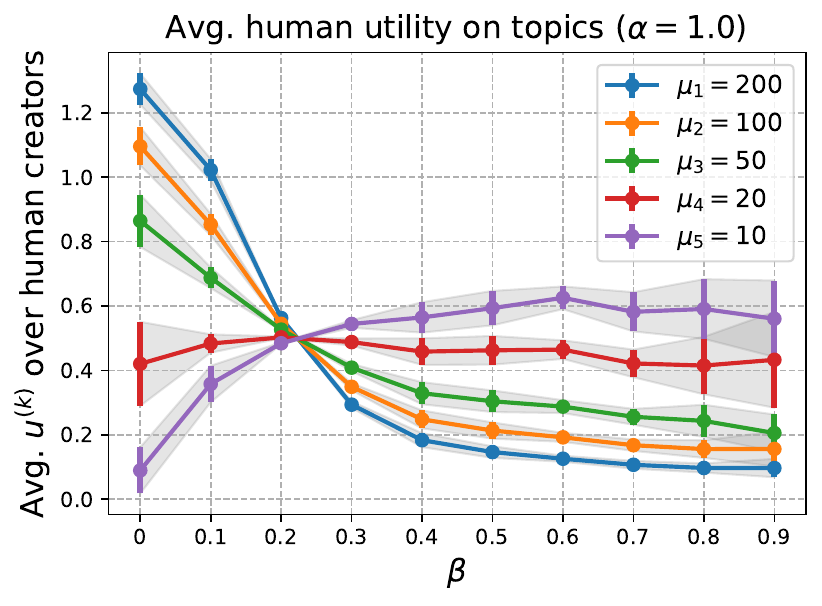}
\caption{$\alpha=1.0$. Left: the occupation ratio on each topic. Right: the average per-topic gain. Cost uniformly sampled from $\U[1.0,10.0]$.}
\label{fig:beta_s_non_1}
\end{figure}

This trend suggests a future where GenAI dominates in easily accessible content areas, encouraging human creators to specialize in niche topics that GenAI is less equipped to handle. Such a shift not only alleviates the competitive pressure on human creators but also potentially enhances their average utility by focusing on areas where they can offer unique value. The detailed results across various settings of $\alpha$ and creator costs further support this observation, as reported in Figure \ref{fig:beta_s_non_4}, \ref{fig:beta_s_non_2}, \ref{fig:beta_s_non_1}. We can see a threshold effect for $\beta$ still exists: beyond a certain point, human creators' preference over topics shifts from popularity to niche. Comparing Figure \ref{fig:beta_s_non_3} with Figure \ref{fig:beta_s_non_4}, \ref{fig:beta_s_non_2}, \ref{fig:beta_s_non_1},  we can see the critical threshold of $\beta$ increases when $\alpha$ becomes smaller and costs becomes larger. This outcome underscores a positive dynamic: GenAI's proliferation in general content domains may lead to a more diversified and specialized content ecosystem, where human creativity is directed towards challenges that GenAI cannot easily replicate.

%% file: conclusion.tex
\section{Conclusion}
In this study, we introduce a competition model that captures the tension between human and GenAI creators in content creation. Our model examines how GenAI influences the equilibrium, including scenarios where creators can choose whether to incorporate GenAI tools. Our analysis reveals that despite the disruptive potential of GenAI, a balanced and desirable equilibrium in the content creation market could be achieved, and how the properties of the market (e.g., its size) and GenAI technology (e.g., its learning capacity) affect the equilibrium. To our best knowledge, this work is the first attempt in developing a game-theoretic framework for understanding the dynamics between human and GenAI creators, offering useful insights into the future of the online content market and shedding light on the evolving landscape of human-AI competition in content creation. The results obtained by our study also open up new avenues of research exploration, for example how to design regulatory interventions to shape the equilibrium (e.g., protect certain type of human creators) or ensure fairness.

\vspace{4mm}
\noindent \textbf{Acknowledgment. } This project is   supported in part by the AI2050 program at Schmidt Sciences (Grant G-24-66104), NSF Award IIS-2213700, IIS-2128019, CCF-2303372,  Army Research Office Award W911NF-23-1-0030, and Office of Naval Research Award N00014-23-1-2802.

%% file: appendix.tex
\appendix
\onecolumn

\section{Proof of Theorem \ref{thm:monotone_G}}\label{app:proof_thm1}
% \begin{equation}\label{eq:ug}
%     u_i(\x_i,\x_{-i})=\sum_{k=1}^K \frac{x_{ik}}{g_k\left(\sum_{j=1}^n x_{jk}\right)}\cdot \mu_k - c_i(\x_i).
% \end{equation}

\begin{proof}
For simplicity we denote $g_k(s)=s^{\gamma_k}+\alpha_k\cdot s^{\beta_k}$. Then the utility function can be expressed as $u_i(\x_i,\x_{-i})=\sum_{k=1}^K\frac{x_{ik}\mu_k}{g_k(s_k)}-c_i(\x_i)$ where $s_k = \sum_{i} x_{ik}$. Our proof starts from a sufficient condition of \citet{rosen1965existence} for a game to be monotone. A game is said to satisfy the \emph{diagonal strict concavity} (DSC) condition if (1) each player has a concave utility function in his own strategy in a convex strategy space; and (2) there exists \emph{some} non-zero parameter $\lambda=(\lambda_1, \cdots, \lambda_n)$ such that the Hessian matrix given by
$$H_{ij}(\x;\lambda)\triangleq\frac{\lambda_i}{2}\frac{\partial^2 u_i(\x)}{\partial \x_i \partial\x_j}+\frac{\lambda_j}{2}\frac{\partial^2 u_j(\x)}{\partial \x_j \partial\x_i}$$
is \emph{strictly} negative-definite. \citet{rosen1965existence} shows that  any game satisfying $\lambda$-DSC condition has a unique pure Nash equilibrium (PNE); such a game is often referred to as  monotone games. 

First of all, since any arbitrarily large $x_{ik}$ would result in a negative utility which is strictly less than the utility obtained at $\bm{0}$, we may without loss of generality let each creator's strategy set be a large hypercube in $\RR^K$, which is a convex set.   Core to our proof is to show that our game $\G$ is $\bm{1}$-DSC under the theorem conditions. Direct calculation shows that

\begin{align}\notag
-[H_{ij}(\x)] &= \sum_{k=1}^K g_k^{-3}\left\{\left(\frac{g_k''g_k}{2}-(g_k')^2\right)
    \begin{bmatrix}
    2 x_1 & x_1+x_2 & \ldots \\
    x_1+x_2 & 2x_2 & \ldots \\
    \vdots & \vdots & \ddots
    \end{bmatrix} + (g_k'g_k)
    \begin{bmatrix}
    2  & 1 & \ldots \\
    1 & 2 & \ldots \\
    \vdots & \vdots & \ddots
    \end{bmatrix}\right\}  \\\notag
    &~~~~+\begin{bmatrix}
    \frac{\partial^2 c_1}{\partial \x_1^2}  & 0 & \ldots \\
    0 & \frac{\partial^2 c_2}{\partial \x_2^2} & \ldots \\
    \vdots & \vdots & \ddots
    \end{bmatrix} \\ \label{eq:31}
    &\triangleq \sum_{k=1}^K H_k+H_0,
\end{align}
where $g_k$ in the above expressions denotes $g_k(s_k)$. We can see that if all the cost functions are strongly convex, the second diagonal matrix $H_0$ in the RHS of Eq. \eqref{eq:31} is strictly positive-definite (PD). Therefore, it suffices to show that (1) for all $k\in [n]$, $H_k$ is always positive semi-definite (PSD), and (2) if $\alpha_k>0$ for some $k\in[n]$, $H_k$ is PD.

Next we omit the subscript $k$ and show that for any function $g(s)$ satisfying certain technical conditions --- specifically, $0\leq 2s(g')^2-sg''g<2g'g$ --- the following matrix $H$ is PD.

\begin{align}\notag
    H &= \left(\frac{g''g}{2}-(g')^2\right)
    \begin{bmatrix}
    2 x_1 & x_1+x_2 & \ldots \\
    x_1+x_2 & 2x_2 & \ldots \\
    \vdots & \vdots & \ddots
    \end{bmatrix} + (g'g)
    \begin{bmatrix}
    2  & 1 & \ldots \\
    1 & 2 & \ldots \\
    \vdots & \vdots & \ddots
    \end{bmatrix} \\ \notag
    & = \left(-\frac{g''g}{2}+(g')^2\right)\begin{bmatrix}
            -2x_1 & -x_1-x_2 & \ldots \\
            -x_1-x_2 & -2x_2 & \ldots \\
            \vdots & \vdots & \ddots
        \end{bmatrix}+\frac{g'g}{s}\left(\sum_{i=1}^nx_i\right)\begin{bmatrix}
        2  & 1 & \ldots \\
        1 & 2 & \ldots \\
        \vdots & \vdots & \ddots
        \end{bmatrix}\\ \notag
    & = \left(-\frac{g''g}{2}+(g')^2\right) \left\{\left(\sum_{i=1}^nx_i\right)\begin{bmatrix}
        2  & 1 & \ldots \\
        1 & 2 & \ldots \\
        \vdots & \vdots & \ddots
        \end{bmatrix}+\begin{bmatrix}
            -2x_1 & -x_1-x_2 & \ldots \\
            -x_1-x_2 & -2x_2 & \ldots \\
            \vdots & \vdots & \ddots
        \end{bmatrix} \right\} \\ \notag 
        & ~~~~+ \left(\frac{g'g}{s}+\frac{g''g}{2}-(g')^2\right)\left(\sum_{i=1}^nx_i\right)\begin{bmatrix}
        2  & 1 & \ldots \\
        1 & 2 & \ldots \\
        \vdots & \vdots & \ddots
        \end{bmatrix}\\ \notag
    & = \left(-\frac{g''g}{2}+(g')^2\right)\begin{bmatrix}
            2\sum_{i\neq 1}x_i & \sum_{i\notin \{1,2\}}x_i & \ldots \\
            \sum_{i\notin \{1,2\}}x_i & 2\sum_{i\neq 2}x_i & \ldots \\
            \vdots & \vdots & \ddots
        \end{bmatrix}+\left(\frac{g'g}{s}+\frac{g''g}{2}-(g')^2\right)\left(\sum_{i=1}^nx_i\right)\begin{bmatrix}
        2  & 1 & \ldots \\
        1 & 2 & \ldots \\
        \vdots & \vdots & \ddots
        \end{bmatrix} , \\  \label{eq:def_M1M2}
    & \triangleq \left(-\frac{g''g}{2}+(g')^2\right)M_1(\x)+\left(\frac{g'g}{s}+\frac{g''g}{2}-(g')^2\right)\left(\sum_{i=1}^nx_i\right)M_2.
\end{align}
To proceed, we need the following Lemma:
    \begin{lemma}\label{lm:psd_M1M2}
        Matrix $M_1,M_2$ defined in Eq. \eqref{eq:def_M1M2} are positive semi-definite (PSD). Moreover, we have $\lambda_{\min}(M_2)=1$.
    \end{lemma}
    \begin{proof}
        To show $M_1(\x)$ is PSD, it suffices to show that for any $\y=(y_1,\cdots,y_n)\in\RR^n$, $\y^{\top}M_1(\x)\y\geq 0$. In fact,
    
        \begin{align*}
            \y^{\top}M_1(\x)\y & = 2\sum_{i=1}^n y_i^2\left(\sum_{j\neq i} x_j \right) + 2\sum_{i<j} y_iy_j\left(\sum_{k\notin \{i,j\}} x_k \right) \\
            & = \sum_{i=1}^n y_i^2\left(\sum_{j\neq i} x_j \right) + \left[\sum_{i=1}^n y_i^2\left(\sum_{j\neq i} x_j \right) + 2\sum_{i<j} y_iy_j\left(\sum_{k\notin \{i,j\}} x_k \right)\right]
            \\
            & = \sum_{i=1}^n y_i^2\left(\sum_{j\neq i} x_j \right) + \sum_{k=1}^n x_k \left[ \sum_{j\neq k}y_j^2 + 2\sum_{i<j,i\neq k,j\neq k} y_iy_j \right]
            \\
            & = \sum_{i=1}^n y_i^2\left(\sum_{j\neq i} x_j \right) + \sum_{i=1}^n x_i\left(\sum_{j\neq i}y_j\right)^2 \geq 0.
        \end{align*}

        On the other hand, $M_2=I+\bm{1}\cdot\bm{1}^{\top}$ and therefore $\lambda_{\min}(M_2)=1$.
    \end{proof}
Based on Lemma \ref{lm:psd_M1M2}, we can see that a sufficient condition for $H$ to be PSD is 
    \begin{equation}
        \left\{ 
            \begin{array}{cc}
                -\frac{g''g}{2}+(g')^2\geq 0,   \\
                \frac{g'g}{s}+\frac{g''g}{2}-(g')^2\geq0,
            \end{array}
        \right.
    \end{equation}
which is equivalent to  
\begin{equation}\label{eq:condition_g}
    0\leq 2s(g')^2-sg''g\leq2g'g.
\end{equation}

When $g(s)=s^{\gamma}+\alpha\cdot s^{\beta}$, we have $g'=\gamma s^{\gamma-1}+\alpha \beta s^{\beta-1}, g''=\gamma(\gamma-1)s^{\gamma-2}$. Condition \eqref{eq:condition_g} is equivalent to 
    \begin{equation}
        \left\{ 
            \begin{array}{cc}
                \gamma(\gamma+1)s^{2\gamma-2}+\alpha^2\beta(\beta+1)s^{2\beta-2}+\alpha(4\beta\gamma-\beta^2+\beta-\gamma^2+\gamma)s^{\gamma+\beta-2}\geq 0,   \\
                \gamma(\gamma-1)s^{2\gamma-2}+\alpha^2\beta(\beta-1)s^{2\beta-2}+\alpha(4\beta\gamma-\beta^2-\beta-\gamma^2-\gamma)s^{\gamma+\beta-2}\leq 0. \label{eq:97}
            \end{array}
        \right.
    \end{equation}

Clearly, the first equation above always holds because $\beta\geq \beta^2, \gamma\geq \gamma^2$. The second equation holds because $\gamma\leq 1, \beta\leq 1$, $\alpha\geq0$ and
$$\beta^2+\beta+\gamma^2+\gamma\geq 2\beta^2+2\gamma^2\geq 4\gamma\beta. $$

In addition, if $\alpha$ is strictly positive (meaning there exists $k\in[n]$ such that $\alpha_k>0$), it holds that $\frac{g'g}{s}+\frac{g''g}{2}-(g')^2>0$ and $H$ is PD.
\end{proof}

% \section{Proof of Corollary \ref{coro:pga_converge}}

% In this section we prove that if all players in the game employ the projected gradient ascent (PGA) algorithm described in Alg. , their joint strategy will eventually converge to the unique PNE. 

% \begin{corollary}\label{coro:pga_converge}
%      If all creators use projected gradient ascent (PGA) algorithm to update their strategies simultaneously, their joint strategy profile converges to the unique PNE of $\G$.
%  \end{corollary}

%  Theorem \ref{coro:pga_converge} is an extension of the convergence result shown in \cite{bravo2018bandit}. To show this result,

\section{Proof of Theorem \ref{lm:order_c}}\label{app:proof_lm_order_c}
We recapitulate the three conclusions of Theorem \ref{lm:order_c} in the following statement: suppose each creator $i$ is equipped with the following utility function
\begin{equation}\label{eq:u_118}
    u_i(x_i,\x_{-i}) = \frac{x_{i}\cdot \mu}{\alpha(\sum_{j=1}^n x_{j})^{\beta}+(\sum_{j=1}^n x_{j})^{\gamma}} - c_ix_i^{\rho}, x_i\in \RR_{\geq 0},
\end{equation}
where $\mu,\alpha\geq0, 0\leq\beta\leq \gamma\leq 1, \rho>1$, and suppose $0<c_1\leq\cdots\leq c_n$. $\x^*=(x_1,\cdots, x_n)$ is the PNE of the game (necessarily unique). Then, the following statements hold:
\begin{enumerate}
    \item $x_1^*\geq \cdots\geq x_n^*$, and $u_1(\x^*)\geq \cdots\geq u_n(\x^*)$.
    \item if the $n$-th creator's cost increases from $c_n$ to $\tilde{c}_n>c_n$ and the new PNE is $\tilde{\x}^*$, we have $u_n(\tilde{\x}^*)<u_n(\x^*)$.
    \item if a new player with cost $c_{n+1}$ joins the competition and $\x'=(x'_1, \cdots, x'_n, x'_{n+1})$ is the new PNE, it holds that 
    $\sum_{i=1}^n x'_i < \sum_{i=1}^n x^*_i$.
\end{enumerate}

The proof of Theorem \ref{lm:order_c} relies on the following Lemma \ref{lm:first_order_pne}, \ref{lm:decreasing_f}, \ref{lm:increasing_h} and \ref{lm:BR_property}. Lemma \ref{lm:first_order_pne} gives the first-order characterization of the PNE of \basegame{} and will be used extensively in our analysis. Lemma \ref{lm:decreasing_f} and \ref{lm:increasing_h} point out monotone properties of different components of the utility functions and their derivatives, and Lemma \ref{lm:BR_property} describes a property of the best response function of any players in \basegame{}.

\begin{lemma}\label{lm:first_order_pne}
The PNE $\x^*$ (necessarily unique) of the game $\G$ specified by utility functions given in Eq. \eqref{eq:u_118} satisfies the following first-order condition for every $  i\in[n]$:

\begin{equation}\label{eq:first_order_117}
    \frac{\partial u_i}{\partial x_i}\bigg|_{x=x^*}=\frac{\mu}{(s^*)^{\gamma}+\alpha \cdot(s^*)^{\beta}} -\frac{\mu x^*_i\cdot[\gamma\cdot(s^*)^{\gamma-1}+\alpha\beta \cdot(s^*)^{\beta-1}]}{[(s^*)^{\gamma}+\alpha \cdot(s^*)^{\beta}]^2} -c_i\rho \cdot(x_i^*)^{\rho-1}=0,
\end{equation}
% \begin{equation}
%         \frac{\mu}{s^{\gamma}+\alpha s^{\beta}} -\frac{\mu x_i\cdot[\gamma s^{\gamma-1}+\alpha\beta s^{\beta-1}]}{[s^{\gamma}+\alpha s^{\beta}]^2} -c_i\rho x_i^{\rho-1}=0, \forall i\in[n],
% \end{equation}
where $s^*=\sum_{i=1}^n x_i^*$.

\end{lemma}
\begin{proof}
We prove the ``if'' and ``only if'' directions separately:
\begin{enumerate}
    \item Suppose $\x^*$ is the PNE of $\G$. Since $\frac{\partial u_i}{\partial x_i}$ is strictly positive when $x_i=0$, we have $x_i^*>0, \forall i\in[n]$. Theorem \ref{thm:monotone_G} suggests that $\G$ is a strictly monotone game and thus $u_i$ is strictly concave in $x_i$ given any fixed $\x_{-i}$. Hence, fixed $\x_{-i}^*$, we know $x_i^*$ is the unique maximizer of a strictly concave function $u_i(x_i, \x^*_{-i})$, which yields the first-order condition Eq. \eqref{eq:first_order_117}.
    \item Suppose $\x^*$ is a solution of Eq. \eqref{eq:first_order_117}, we show $\x^*$ must be the unique PNE of $\G$. First of all, the solution of Eq. \eqref{eq:first_order_117} must exist since we already showed that the unique PNE of $\G$ satisfies Eq. \eqref{eq:first_order_117}. If there $\tilde{\x}^*$ is another solution of Eq. \eqref{eq:first_order_117}, since $u_i$ is strictly concave in $x_i$, from the first-order condition we have $\tilde{x}_i$ is the maximizer of $u_i$ given $\tilde{\x}_{-i}$, which means $\tilde{\x}^*$ is also a PNE of $\G$, contradicting the uniqueness of $\x^*$.
\end{enumerate}

\end{proof}

\begin{lemma}\label{lm:decreasing_f}
For any constants $\alpha\geq 0, c_0 >0$ and $\beta,\gamma\in [0,1]$, the following function is monotonically decreasing when $s>2c_0$:
$$f(s)=\frac{1}{s^{\gamma}+\alpha s^{\beta}} -\frac{c_0\cdot[\gamma s^{\gamma-1}+\alpha\beta s^{\beta-1}]}{[s^{\gamma}+\alpha s^{\beta}]^2}, s> 2c_0.$$
\end{lemma}
\begin{proof}
    We compute the first-order derivative of $f$ directly. 
    \begin{align}\notag
        (s^{\gamma}+\alpha s^{\beta})^3f'(s)= &\gamma s^{2\gamma-2}(c_0\gamma+c_0-s) +\alpha^2\beta s^{2\beta-2}(c_0\beta+c_0-s) \\ \label{eq:420}
        &+\alpha s^{\gamma+\beta-2}\left[\gamma(c_0\gamma+c_0-s)+\beta(c_0\beta+c_0-s)-2c_0(\beta-\gamma)^2\right].
    \end{align}
    Since $\beta, \gamma \leq 1$ and $s>2c_0$, it holds that
    $c_0\gamma+c_0-s<0$ and $c_0\beta+c_0-s<0$. Therefore, each term of the RHS of Eq. \eqref{eq:420} is strictly negative. Hence, we conclude $f'(s)<0, \forall s>2c_0$ and the claim holds true.
\end{proof}

\begin{lemma}\label{lm:increasing_h}
For any constants $\alpha\geq 0, B >0$ and $\beta,\gamma\in [0,1]$, the following two functions are both monotonically increasing w.r.t. $x$:

\begin{align*}
    h_1(x) & =\frac{x}{(B+x)^{\gamma}+\alpha (B+x)^{\beta}}, \qquad \text{where } x\in \RR_{+}, \\ 
h_2(x) &=\frac{x^2(\gamma(B+x)^{\gamma-1}+\alpha\beta(B+x)^{\beta-1})}{[(B+x)^{\gamma}+\alpha (B+x)^{\beta}]^2}, \qquad \text{where } x> 1-B. 
\end{align*}

\end{lemma}
\begin{proof}
Direct calculation shows that
$$h'_1(x)=[(B+x)^{\gamma}+\alpha (B+x)^{\beta}]^{-2}\cdot \left((B+x)^{\gamma-1}(B+(1-\gamma) x)+\alpha(B+x)^{\beta-1}(B+(1-\beta)x)\right)$$
which is strictly positive. So $h_1(x)$ is monotonically increasing. 
Morepver, we have 
\begin{align}\notag
    & x^{-1}[(B+x)^{\gamma}+\alpha (B+x)^{\beta}]^{3}h'_2(x) \\  \notag
    = \, \,  & \gamma\left[2 B+(1-\gamma)x\right]\cdot(B+x)^{2\gamma-2}\\  \notag&+\alpha\left[(\beta+\gamma)(2B+2x-1)+\beta^2+\gamma^2-4\beta\gamma\right]\cdot(B+x)^{\gamma+\beta-2} \\\label{eq:224}
    &+\alpha^2\beta\left[2B+2x-\beta-1\right]\cdot(B+x)^{2\beta-2} 
\end{align}

Since $B+x> 1$, $\beta,\gamma \leq 1$, we have 
\begin{align*}
    & (\beta+\gamma)(2B+2x-1)+\beta^2+\gamma^2-4\beta\gamma \\
    \geq  \, \,  &(\beta+\gamma)+\beta^2+\gamma^2-4\beta\gamma \\ 
    \geq \, \,   & 2\sqrt{\beta\gamma}+2\beta\gamma-4\beta\gamma \geq 0, & \mbox{since $\beta\gamma\leq 1$}
\end{align*}
and $2B+2x-\beta-1> 1-\beta\geq 0$.
Therefore, the RHS of Eq. \eqref{eq:224} is strictly positive and it holds that $h'_2(x)>0$.

\end{proof}

\begin{lemma}\label{lm:BR_property}
    In any \basegame{}$(\alpha,\beta,\gamma,\mu,\rho,\{c_i\})$, if we fix the $n$-th player's pure strategy $x_n=y_n$, then 
    \begin{enumerate}
        \item the sub-game for the remaining $n-1$ players admits a unique PNE $(y_1,\cdots,y_{n-1})$;
        \item when $n, \mu$ are sufficiently large and $1< \rho\leq 2$, it holds that $\max_{i=1}^{n-1}\{y_i\} < \frac{1}{2}\sum_{i=1}^{n-1} y_i$.
        \item let $BR(y_n)=\sum_{i=1}^{n-1}y_{i}$ be the function that maps player $n$'s arbitrary strategy $y_n$ to the sum of other players' equilibrium strategies. Then $BR(y_n)$ is strictly decreasing in $y_n\in\RR_{+}$.
    \end{enumerate} 
\end{lemma}

\begin{proof}
 The first claim is a corollary of Theorem \ref{thm:monotone_G}. In the proof of Theorem \ref{thm:monotone_G}, we know that the $n$-player game \basegame{}$(\alpha,\beta,\gamma,\mu,\rho,\{c_i\}_{i=1}^n)$ is strictly monotone and satisfies $\bm{1}$-DSC, i.e., the game Hessian
$$H_{ij}(\x;\bm{1})\triangleq\frac{1}{2}\frac{\partial^2 u_i(\x)}{\partial \x_i \partial\x_j}+\frac{1}{2}\frac{\partial^2 u_j(\x)}{\partial \x_j \partial\x_i}$$ is strictly negative-definite. When the $n$-th player's strategy is fixed, the Hessian of the sub-game for the remaining $(n-1)$ players is given by the $(n-1)\times(n-1)$ submatrix of $H_{ij}(\x;\bm{1})$ by excluding its $n$-th row and $n$-th column and therefore is still strictly negative-definite. Hence, the sub-game is also $\bm{1}$-DSC and admits a unique PNE.

Now we show the second claim. Since we without loss of generality assume $0<c_1\leq c_2\leq \cdots \leq c_n$, from the first-order condition Eq. \eqref{eq:first_order_117} it follows that $y_1\geq y_2\geq \cdots \geq y_{n-1}$. Let $s=\sum_{i=1}^{n-1}y_i$, it suffices to show $s>2y_1$.

From the definition of PNE we have 
\begin{equation}\label{eq:169}
    \frac{\mu y_1}{(s+y_n)^{\gamma}+\alpha (s+y_n)^{\beta}}-c_1y_1^{\rho}=u_1(y_1,\y_{-1})\geq u_1(0,\y_{-1})=0,
\end{equation}
Since $1<\rho\leq 2$, Eq. \eqref{eq:169} implies 
    \begin{equation*}
        y_1\leq \left(\frac{\mu}{c_1((s+y_n)^{\gamma}+\alpha (s+y_n)^{\beta})}\right)^{1/(\rho-1)}<\left(\frac{\mu}{c_1 s^{\gamma}}\right)^{1/(\rho-1)}.
    \end{equation*}
    Hence, it is sufficient to show $s>2\left(\frac{\mu}{c_1 s^{\gamma}}\right)^{1/(\rho-1)}$, which is equivalent to
    \begin{equation}\label{eq:161}
        s>2^{\frac{\rho-1}{\rho+\gamma-1}}\cdot\left(\frac{\mu}{c_1}\right)^{\frac{1}{\rho+\gamma-1}}.
    \end{equation}
    We prove Eq. \eqref{eq:161} by contradiction. Suppose $s\leq2^{\frac{\rho-1}{\rho+\gamma-1}}\cdot\left(\frac{\mu}{c_1}\right)^{\frac{1}{\rho+\gamma-1}}$ and denote $s_0=s+y_n$. Since $y_n$ is a fixed constant and $\mu$ is sufficiently large, we have $s_0\leq3^{\frac{\rho-1}{\rho+\gamma-1}}\cdot\left(\frac{\mu}{c_1}\right)^{\frac{1}{\rho+\gamma-1}}$. Then from $y_1\geq y_2\geq \cdots \geq y_{n-1}$ we have 
    \begin{equation}\label{eq:181}
        y_{n-1} \leq \frac{s}{n-1}\leq \frac{2^{\frac{\rho-1}{\rho+\gamma-1}}}{n-1}\cdot\left(\frac{\mu}{c_1}\right)^{\frac{1}{\rho+\gamma-1}}.
    \end{equation}

    Since $(y_1,\cdots,y_{n-1})$ is the PNE of the sub-game, it also satisfies Eq. \eqref{eq:first_order_117} in Lemma \ref{lm:first_order_pne}. For player $n-1$ we have 
    \begin{align}\notag
        0&= \frac{\partial u_{n-1}}{\partial x_{n-1}}\Big|_{x_{n-1}=y_{n-1}} \\ \notag
        &=\frac{\mu}{s_0^{\gamma}+\alpha s_0^{\beta}} -\frac{\mu y_{n-1} (\gamma s_0^{\gamma-1}+\alpha\beta s_0^{\beta-1})}{(s_0^{\gamma}+\alpha s_0^{\beta})^2} -\rho c_{n-1} y_{n-1}^{\rho-1} \\ \notag
        &\geq \frac{ \mu}{s_0^{\gamma}+\alpha s_0^{\beta}} -\frac{\mu s_0( s_0^{\gamma-1}+\alpha_0^{\beta-1})}{(n-1)(s_0^{\gamma}+\alpha s_0^{\beta})^2}  -\rho c_{n-1} y_{n-1}^{\rho-1} &\mbox{since $\gamma,\beta\leq 1$}\\ \notag
        &\geq \frac{ \mu}{s_0^{\gamma}+\alpha s_0^{\beta}}\left(1-\frac{1}{n-1}\right) -\frac{\rho c_{n-1} \cdot 2^{\frac{(\rho-1)^2}{\rho+\gamma-1}}}{(n-1)^{\rho-1}}\cdot\left(\frac{\mu}{c_1}\right)^{\frac{\rho-1}{\rho+\gamma-1}} &\mbox{from Eq. \eqref{eq:181}}\\ \notag
        & \geq \frac{\mu}{s_0^{\gamma}(1+\alpha)}\left(\frac{n-2}{n-1}\right) -\frac{\rho c_{n-1}\cdot2^{\frac{(\rho-1)^2}{\rho+\gamma-1}}}{(n-1)^{\rho-1}c_1^{\frac{\rho-1}{\rho+\gamma-1}}}\cdot\mu^{\frac{\rho-1}{\rho+\gamma-1}} &\mbox{since $\beta\leq\gamma\leq 1$} \\\label{eq:178}
        & \geq \left[\frac{3^{\frac{-\gamma(\rho-1)}{\rho+\gamma-1}}c_{n-1}^{\frac{\gamma}{\rho+\gamma-1}}}{1+\alpha}\left(\frac{n-2}{n-1}\right) - \frac{\rho c_{n-1}\cdot2^{\frac{(\rho-1)^2}{\rho+\gamma-1}}}{(n-1)^{\rho-1}c_1^{\frac{\rho-1}{\rho+\gamma-1}}}\right] \cdot\mu^{\frac{\rho-1}{\rho+\gamma-1}} &\mbox{since $s_0\leq3^{\frac{\rho-1}{\rho+\gamma-1}} \left(\frac{\mu}{c_1}\right)^{\frac{1}{\rho+\gamma-1}}$} \\ \label{eq:194}
        &>0.
    \end{align}
    The last Eq. \eqref{eq:194} holds because when $n\rightarrow +\infty$, $1/n^{\rho-1}\rightarrow 0$ and therefore the expression in the brackets of Eq. \eqref{eq:178} is strictly positive when $n$ is large enough. Eq. \eqref{eq:194} draws a contradiction to the first-order equilibrium condition which completes the proof.

Next we prove the Third claim. Let $y'_n>y_n$ be another strategy of player $n$ and we show that $BR(y'_n)<BR(y_n)$. Define variable $s=\sum_{i=1}^{n-1} x_i$, and constants $s^*=BR(y_n)=\sum_{i=1}^{n-1} y_i, s'=BR(y'_n)=\sum_{i=1}^{n-1} y'_i$. From Lemma \ref{lm:first_order_pne} we know  $(y_1, \cdots, y_{n-1})$ is the unique solution for the following equalities for all $i=1,2,\cdots,n-1$: 
\begin{equation}\label{eq:165}
        \frac{\mu}{(s+y_n)^{\gamma}+\alpha (s+y_n)^{\beta}} -\frac{\mu x_i\cdot[\gamma (s+y_n)^{\gamma-1}+\alpha\beta (s+y_n)^{\beta-1}]}{[(s+y_n)^{\gamma}+\alpha (s+y_n)^{\beta}]^2} -c_i\rho x_i^{\rho-1}=0. 
\end{equation}
Similarly, $(y'_1, \cdots, y'_{n-1})$ is the unique solution for the following equalities for all $i=1,2,\cdots,n-1$: 
\begin{equation}\label{eq:170}
        \frac{\mu}{(s+y'_n)^{\gamma}+\alpha (s+y'_n)^{\beta}} -\frac{\mu x_i\cdot[\gamma (s+y'_n)^{\gamma-1}+\alpha\beta (s+y'_n)^{\beta-1}]}{[(s+y'_n)^{\gamma}+\alpha (s+y'_n)^{\beta}]^2} -c_i\rho x_i^{\rho-1}=0. 
\end{equation}
Next we prove the claim by contradiction. If the claim is not true, i.e., $s'\geq s^*$. Then there must exist $j\in[n-1]$ such that $y'_j\geq y_j$. And also from the second claim of Lemma \ref{lm:BR_property} we have $s^*>2y_j, s'>2y'_j$ when $n$ is large. As a result, it holds that 

\begin{align*}
   & c_j\rho y_j^{\rho-1} \\ =&\frac{\mu}{(s^*+y_n)^{\gamma}+\alpha (s^*+y_n)^{\beta}} -\frac{\mu y_j\cdot[\gamma (s^*+y_n)^{\gamma-1}+\alpha\beta (s^*+y_n)^{\beta-1}]}{[(s^*+y_n)^{\gamma}+\alpha (s^*+y_n)^{\beta}]^2} & \mbox{from Eq. \eqref{eq:165}}\\  
   \geq & \frac{\mu}{(s'+y_n)^{\gamma}+\alpha (s'+y_n)^{\beta}} -\frac{\mu y_j\cdot[\gamma (s'+y_n)^{\gamma-1}+\alpha\beta (s'+y_n)^{\beta-1}]}{[(s'+y_n)^{\gamma}+\alpha (s'+y_n)^{\beta}]^2}  & \mbox{from Lemma \ref{lm:decreasing_f}}\\
   \geq & \frac{\mu}{(s'+y_n)^{\gamma}+\alpha (s'+y_n)^{\beta}} -\frac{\mu y'_j\cdot[\gamma (s'+y_n)^{\gamma-1}+\alpha\beta (s'+y_n)^{\beta-1}]}{[(s'+y_n)^{\gamma}+\alpha (s'+y_n)^{\beta}]^2} & \mbox{since $y'_j\geq y_j$}\\
    > & \frac{\mu}{(s'+y'_n)^{\gamma}+\alpha (s'+y'_n)^{\beta}} -\frac{\mu y'_j\cdot[\gamma (s'+y'_n)^{\gamma-1}+\alpha\beta (s'+y'_n)^{\beta-1}]}{[(s'+y'_n)^{\gamma}+\alpha (s'+y'_n)^{\beta}]^2}  & \mbox{from Lemma \ref{lm:decreasing_f}}\\
   = & c_j\rho (y'_j)^{\rho-1}. & \mbox{from Eq. \eqref{eq:170}}
\end{align*}

Therefore, from $\rho>1$ we obtain $y_j>y'_j$, which contradicts the fact that $y'_j\geq y_j$. Hence, our third claim in Lemma \ref{lm:BR_property} holds.

\end{proof}

Now we are prepared to prove Theorem \ref{lm:order_c}. For simplicity we omit the superscript $*$ and simply use $\x^*=(x_1,\cdots, x_n)$ to refer to $\G$'s PNE. Let $s=\sum_{i=1}^n x_i$. 

\begin{proof}\textbf{[of Theorem \ref{lm:order_c}]}
\begin{enumerate}
    \item Claim-1: For any $1\leq i<j\leq n$, if $x_i<x_j$, we have 
        \begin{align*}
            0 &= \left(\frac{\partial u_i}{\partial x_i}-\frac{\partial u_j}{\partial x_j}\right)\bigg|_{x=x^*} \\ 
            &= \frac{\mu (x_j- x_i)\cdot[\gamma s^{\gamma-1}+\alpha\beta s^{\beta-1}]}{[s^{\gamma}+\alpha s^{\beta}]^2} + \rho (c_j x_j^{\rho-1}-c_i x_i^{\rho-1}) > 0,
        \end{align*}
        which is a contradiction. As a result, for any $1\leq i<j\leq n$ it holds that $x_i\geq x_j$. Moreover, from Eq. \eqref{eq:first_order_117} we have 
    \begin{equation}\label{eq:156}
        c_i\rho x_i^{\rho-1}=\frac{\mu }{s^{\gamma}+\alpha s^{\beta}}-\frac{\mu x_i\cdot[\gamma s^{\gamma-1}+\alpha\beta s^{\beta-1}]}{[s^{\gamma}+\alpha s^{\beta}]^2}.
    \end{equation}
    Substitute Eq. \eqref{eq:156} into Eq. \eqref{eq:u_118} we obtain 
        \begin{equation}\label{eq:u_160}
            u_i(x_i,\x_{-i}) = \frac{x_{i}\cdot \mu}{s^{\gamma}+\alpha s^{\beta}}\left(1-\frac{1}{\rho}\right) + \frac{\mu x_i^2\cdot[\gamma s^{\gamma-1}+\alpha\beta s^{\beta-1}]}{\rho[s^{\gamma}+\alpha s^{\beta}]^2}.
        \end{equation}
    Therefore, for any $1\leq i<j\leq n$ we have  
    \begin{small}
    \begin{equation*}
        u_i(x_i,\x_{-i})-u_j(x_j,\x_{-j}) = (x_{i}-x_j)\cdot\left(\frac{\mu}{s^{\gamma}+\alpha s^{\beta}}\cdot\frac{\rho-1}{\rho} + \frac{\mu (x_i+x_j)\cdot[\gamma s^{\gamma-1}+\alpha\beta s^{\beta-1}]}{\rho[s^{\gamma}+\alpha s^{\beta}]^2}\right)\geq 0,
    \end{equation*}
    \end{small}
    because $x_i\geq x_j$ and $\rho\geq 1$.
    
    \item Claim-2: Let $s_0=\sum_{i=1}^n \tilde{x}_i$. From Eq. \eqref{eq:first_order_117} we have 
    \begin{align}\label{eq:176}
        & \frac{\mu}{s^{\gamma}+\alpha s^{\beta}} -\frac{\mu x_n\cdot[\gamma s^{\gamma-1}+\alpha\beta s^{\beta-1}]}{[s^{\gamma}+\alpha s^{\beta}]^2} -c_n\rho x_n^{\rho-1}=0, \\ 
        & \frac{\mu}{s_0^{\gamma}+\alpha s_0^{\beta}} -\frac{\mu \tilde{x}_n\cdot[\gamma s_0^{\gamma-1}+\alpha\beta s_0^{\beta-1}]}{[s_0^{\gamma}+\alpha s_0^{\beta}]^2} -\tilde{c}_n\rho \tilde{x}_n^{\rho-1}=0.
    \end{align}
    First we show $\tilde{x}_n<x_n$. From Lemma \ref{lm:BR_property} we know if we define $(y_1,\cdots,y_{n-1})$ as the best-response mapping of player from $1$ to $n-1$ given player $n$'s pure strategy $y_n$, the function $F(y_n)=\sum_{i=1}^{n-1} y_i$ is well-defined and strictly decreasing. Note that $s=F(x_n)+x_n$, the LHS of Eq. \eqref{eq:176} can be rewritten as 
    \begin{align*}
        G(x_n;c_n)&\triangleq\frac{\mu}{(F(x_n)+x_n)^{\gamma}+\alpha (F(x_n)+x_n)^{\beta}} \\ &~~~~-\frac{\mu x_n\cdot[\gamma (F(x_n)+x_n)^{\gamma-1}+\alpha\beta (F(x_n)+x_n)^{\beta-1}]}{[(F(x_n)+x_n)^{\gamma}+\alpha (F(x_n)+x_n)^{\beta}]^2} -c_n\rho x_n^{\rho-1}.
    \end{align*}
    From the existence and uniqueness of $\G$'s PNE we know that $x=x_n$ is the unique root of equation $G(x;c_n)=0$ and $x=\tilde{x}_n$ is the unique root of equation $G(x;\tilde{c}_n)=0$. 
    
    Next, we show $\tilde{x}_n<x_n$. On the one hand, since $F(\cdot)>0$, it holds that $G(0;\tilde{c}_n)>0$. On the other hand, since $\tilde{c}_n>c_n$, it holds that 
    \begin{equation}
        G(x_n;\tilde{c}_n)=G(x_n;c_n)+c_n\rho x_n^{\rho-1}-\tilde{c}_n\rho x_n^{\rho-1}<0.
    \end{equation}
    As a result, the continuous function $G(x;\tilde{c}_n)$ has a root within the interval $(0, x_n)$. Since $\tilde{x}_n$ is its unique root, we obtain $\tilde{x}_n<x_n$. 
    
    From Eq. \eqref{eq:u_160} we have 
    \begin{align}\label{eq:192}
        & u_n(x_n,\x_{-n}) = \frac{x_n\cdot \mu}{s^{\gamma}+\alpha s^{\beta}}\left(1-\frac{1}{\rho}\right) + \frac{\mu x_n^2\cdot[\gamma s^{\gamma-1}+\alpha\beta s^{\beta-1}]}{\rho[s^{\gamma}+\alpha s^{\beta}]^2},\\ \label{eq:193}
        & u_n(\tilde{x}_n,\tilde{\x}_{-n}) = \frac{\tilde{x}_n\cdot \mu}{s_0^{\gamma}+\alpha s_0^{\beta}}\left(1-\frac{1}{\rho}\right) + \frac{\mu \tilde{x}_n^2\cdot[\gamma s_0^{\gamma-1}+\alpha\beta s_0^{\beta-1}]}{\rho[s_0^{\gamma}+\alpha s_0^{\beta}]^2}.
    \end{align}
    From Eq. \eqref{eq:192} and \eqref{eq:193}, a sufficient condition for $u_n(x_n,\x_{-n})>u_n(\tilde{x}_n,\tilde{\x}_{-n})$ to hold is the following two inequalities:
    \begin{align}\label{eq:197}
        & \frac{x_n}{s^{\gamma}+\alpha s^{\beta}} > \frac{\tilde{x}_n}{s_0^{\gamma}+\alpha s_0^{\beta}},\\ \label{eq:198}
        & \left(\frac{x_n}{s^{\gamma}+\alpha s^{\beta}}\right)^2\cdot[\gamma s^{\gamma-1}+\alpha\beta s^{\beta-1}] > \left(\frac{\tilde{x}_n}{s_0^{\gamma}+\alpha s_0^{\beta}}\right)^2\cdot[\gamma s_0^{\gamma-1}+\alpha\beta s_0^{\beta-1}].
    \end{align}
    Next we prove Eq. \eqref{eq:197} and \eqref{eq:198}. Let $A=F(x_n), B=F(\tilde{x}_n)$ and from Lemma \ref{lm:BR_property} we have $0<A<B$. Then it holds that
    \begin{align*}
        \frac{x_n}{s^{\gamma}+\alpha s^{\beta}} =& \frac{x_n}{(A+x_n)^{\gamma}+\alpha (A+x_n)^{\beta}} \\ 
        > & \frac{x_n}{(B+x_n)^{\gamma}+\alpha (B+x_n)^{\beta}} & \mbox{since $A< B$} \\
        > & \frac{\tilde{x}_n}{(B+\tilde{x}_n)^{\gamma}+\alpha (B+\tilde{x}_n)^{\beta}} & \mbox{from Lemma \ref{lm:increasing_h}} \\= &\frac{\tilde{x}_n}{s_0^{\gamma}+\alpha s_0^{\beta}}.
    \end{align*}
    Similarly, we can obtain 
    \begin{align*}
        &\left(\frac{x_n}{s^{\gamma}+\alpha s^{\beta}}\right)^2\cdot[\gamma s^{\gamma-1}+\alpha\beta s^{\beta-1}] \\ =& \frac{x_n^2(\gamma(A+x_n)^{\gamma-1}+\alpha\beta(A+x_n)^{\beta-1})}{[(A+x_n)^{\gamma}+\alpha (A+x_n)^{\beta}]^2} \\ 
        \geq& \frac{x_n^2(\gamma(B+x_n)^{\gamma-1}+\alpha\beta(B+x_n)^{\beta-1})}{[(B+x_n)^{\gamma}+\alpha (B+x_n)^{\beta}]^2} & \mbox{since $A\leq B$, and $\beta,\gamma\leq 1$}\\ 
        > & \frac{\tilde{x}_n^2(\gamma(B+x_n)^{\gamma-1}+\alpha\beta(B+\tilde{x}_n)^{\beta-1})}{[(B+\tilde{x}_n)^{\gamma}+\alpha (B+\tilde{x}_n)^{\beta}]^2} & \mbox{from Lemma \ref{lm:increasing_h}} \\ 
        = &\left(\frac{\tilde{x}_n}{s_0^{\gamma}+\alpha s_0^{\beta}}\right)^2\cdot[\gamma s_0^{\gamma-1}+\alpha\beta s_0^{\beta-1}].
    \end{align*}
    Hence, we conclude $u_n(x_n,\x_{-n})>u_n(\tilde{x}_n,\tilde{\x}_{-n})$.
    \item Claim-3: Define variable $s=\sum_{i=1}^n x_i$, and constants $s^*=\sum_{i=1}^n x^*_i, s'=\sum_{i=1}^n x'_i$. From Lemma \ref{lm:first_order_pne} we know  $(x^*_1, \cdots, x^*_n)$ is the unique solution of the following equation system
\begin{equation}\label{eq:423}
        \frac{\mu}{s^{\gamma}+\alpha s^{\beta}} -\frac{\mu x_i\cdot[\gamma s^{\gamma-1}+\alpha\beta s^{\beta-1}]}{[s^{\gamma}+\alpha s^{\beta}]^2} -c_i\rho x_i^{\rho-1}=0, i=1,2,\cdots,n.
\end{equation}
Obviously, for any $i$, $x_i>0,x'_i>0$ must hold because the LHS of Eq. \eqref{eq:423} is strictly positive when $x_i=0$. Denote the constant $\epsilon=x'_{n+1}>0$, then similarly we have that $(x'_1, \cdots, x'_n)$ is the unique solution of 
\begin{equation}\label{eq:427}
        \frac{\mu}{(s+\epsilon)^{\gamma}+\alpha (s+\epsilon)^{\beta}} -\frac{\mu x_i\cdot[\gamma (s+\epsilon)^{\gamma-1}+\alpha\beta (s+\epsilon)^{\beta-1}]}{[(s+\epsilon)^{\gamma}+\alpha (s+\epsilon)^{\beta}]^2} -c_i\rho x_i^{\rho-1}=0, i=1,2,\cdots,n.
\end{equation}
Next we prove the claim by contradiction. Suppose Eq. \eqref{eq:414} is not true, i.e., $s'\geq s^*$. Then there must exist $j\in[n]$ such that $x'_j\geq x^*_j$. And also from Lemma \ref{lm:max_x} we have $s^*>2x^*_j, s'>2x'_j$ when $n$ is large. As a result, it holds that
\begin{align*}
   & c_j\rho (x^*_j)^{\rho-1} \\=&\frac{\mu}{(s^*)^{\gamma}+\alpha (s^*)^{\beta}} -\frac{\mu x^*_j\cdot[\gamma (s^*)^{\gamma-1}+\alpha\beta (s^*)^{\beta-1}]}{[(s^*)^{\gamma}+\alpha (s^*)^{\beta}]^2} & \mbox{from Eq. \eqref{eq:423}}\\  
   \geq & \frac{\mu}{(s')^{\gamma}+\alpha (s')^{\beta}} -\frac{\mu x^*_j\cdot[\gamma (s')^{\gamma-1}+\alpha\beta (s')^{\beta-1}]}{[(s')^{\gamma}+\alpha (s')^{\beta}]^2} & \mbox{from Lemma \ref{lm:decreasing_f}}\\
   \geq & \frac{\mu}{(s')^{\gamma}+\alpha (s')^{\beta}} -\frac{\mu x'_j\cdot[\gamma (s')^{\gamma-1}+\alpha\beta (s')^{\beta-1}]}{[(s')^{\gamma}+\alpha (s')^{\beta}]^2} & \mbox{since $x'_j\geq x^*_j$}\\
   > & \frac{\mu}{(s'+\epsilon)^{\gamma}+\alpha (s'+\epsilon)^{\beta}} -\frac{\mu x'_j\cdot[\gamma (s'+\epsilon)^{\gamma-1}+\alpha\beta (s'+\epsilon)^{\beta-1}]}{[(s'+\epsilon)^{\gamma}+\alpha (s'+\epsilon)^{\beta}]^2} & \mbox{from Lemma \ref{lm:decreasing_f}}\\ 
   = & c_j\rho (x'_j)^{\rho-1}. & \mbox{from Eq. \eqref{eq:427}}
\end{align*}

Therefore, from $\rho>1$ we obtain $x^*_j>x'_j$, which contradicts the fact that $x'_j\geq x^*_j$. Hence, Eq. \eqref{eq:414} holds.
\end{enumerate}
\end{proof}

\section{Proof of Proposition \ref{lm:cost_bound}}\label{app:proof_lm_cost_bound}
\begin{proof}
From Eq. \eqref{eq:first_order_117} in Lemma \ref{lm:first_order_pne} it holds that
    \begin{align*}
        \frac{c_i\rho}{\mu}\cdot x_i^{\rho-1}=\frac{1}{s^{\gamma}+\alpha s^{\beta}}-\frac{x_i\cdot[\gamma s^{\gamma-1}+\alpha\beta s^{\beta-1}]}{[s^{\gamma}+\alpha s^{\beta}]^2}<\frac{1}{s^{\gamma}+\alpha s^{\beta}}
        .
    \end{align*}
Hence, $$c_i x_i^{\rho}< \frac{1}{\rho}\cdot \frac{x_i\mu}{s^{\gamma}+\alpha s^{\beta}}.$$ On the other hand, from Theorem \ref{lm:order_c} we know $x_1\geq x_2\cdots\geq x_n$ and therefore $x_i\geq \frac{s}{i}$. Therefore, for any $1<i\leq n$ we have 
    \begin{align}\notag
        \frac{c_i\rho}{\mu}\cdot x_i^{\rho-1}&=\frac{1}{s^{\gamma}+\alpha s^{\beta}}-\frac{x_i\cdot[\gamma s^{\gamma-1}+\alpha\beta s^{\beta-1}]}{[s^{\gamma}+\alpha s^{\beta}]^2} \\\notag
        &\geq \frac{1}{s^{\gamma}+\alpha s^{\beta}}-\frac{x_i\cdot[s^{\gamma-1}+\alpha s^{\beta-1}]}{[s^{\gamma}+\alpha s^{\beta}]^2}&\mbox{since $\beta, \gamma\leq 1$}\\\notag
        &>\frac{1}{s^{\gamma}+\alpha s^{\beta}}-\frac{s\cdot[ s^{\gamma-1}+\alpha s^{\beta-1}]}{i\cdot[s^{\gamma}+\alpha s^{\beta}]^2}\\ \label{eq:256}
        &=\frac{i-1}{i}\cdot \frac{1}{s^{\gamma}+\alpha s^{\beta}}\geq \frac{1}{2(s^{\gamma}+\alpha s^{\beta})}.
    \end{align}

For $i=1$, Lemma \ref{lm:max_x} suggests that $x_1>\frac{s}{2}$ holds for sufficiently large $n$ and $\mu$. In this case, Eq. \eqref{eq:256} also holds for $i=1$. As a result, we obtain that for any $i$, 
\begin{equation*}
    \frac{1}{2\rho}\cdot\frac{x_i\cdot \mu}{s^{\gamma}+\alpha s^{\beta}}<c_i x_i^{\rho} < \frac{1}{\rho}\cdot \frac{x_i\cdot \mu}{s^{\gamma}+\alpha s^{\beta}}.
\end{equation*}

% From the proof above, we can actually derive a stronger version of Lemma \ref{lm:cost_bound} for creators except the one ($i=1$) with the lowest cost:
% \begin{equation}\label{eq:stronger_lm_cost}
%     \frac{i-1}{i\rho}\cdot\frac{x_i\cdot \mu}{s^{\gamma}+\alpha s^{\beta}}<c_i x_i^{\rho} < \frac{1}{\rho}\cdot \frac{x_i\cdot \mu}{s^{\gamma}+\alpha s^{\beta}},\forall i\geq 2.
% \end{equation}
% Eq. \eqref{eq:stronger_lm_cost} adds little value to the message we convey in Lemma \ref{lm:cost_bound}, but is of technical importance and will be used in the proof of Theorem \ref{thm:PNE_extend}.

\end{proof}

\section{Proof of Theorem \ref{thm:asym_s} }\label{app:proof_thm2}

As a preparation, we need the following technical lemma:

\begin{lemma}\label{lm:max_x}
If $n, \mu$ are sufficiently large and $1<\rho\leq 2$, it holds that $\max_{i=1}^n\{x_i\} < \frac{s}{2}$.
\end{lemma}

\begin{proof}
Since we without loss of generality assume $0<c_1\leq c_2\leq \cdots \leq c_n$, from the first-order condition Eq. \eqref{eq:first_order_117} it follows that $x_1\geq x_2\geq \cdots \geq x_n$. Therefore it suffices to show $s>2x_1$.

From the definition of PNE we have 
\begin{equation}\label{eq:150}
    \frac{\mu x_1}{s^{\gamma}+\alpha s^{\beta}}-c_1x_1^{\rho}=u_1(x_1,\x_{-1})\geq u_1(0,\x_{-1})=0,
\end{equation}
Since $1<\rho\leq 2$, Eq. \eqref{eq:150} implies 
    \begin{equation*}
        x_1\leq \left(\frac{\mu}{c_1(s^{\gamma}+\alpha s^{\beta})}\right)^{1/(\rho-1)}.
    \end{equation*}
    Hence, it is sufficient to show $s>2\left(\frac{\mu}{c_1 s^{\gamma}}\right)^{1/(\rho-1)}$, which is equivalent to
    \begin{equation}\label{eq:162}
        s>2^{\frac{\rho-1}{\rho+\gamma-1}}\cdot\left(\frac{\mu}{c_1}\right)^{\frac{1}{\rho+\gamma-1}},
    \end{equation}
    We prove Eq. \eqref{eq:162} by contradiction. Suppose $s\leq2^{\frac{\rho-1}{\rho+\gamma-1}}\cdot\left(\frac{\mu}{c_1}\right)^{\frac{1}{\rho+\gamma-1}}$, then from $x_1\geq x_2\geq \cdots \geq x_n$ we have 
    \begin{equation}\label{eq:166}
        x_n \leq \frac{s}{n}\leq \frac{2^{\frac{\rho-1}{\rho+\gamma-1}}}{n}\cdot\left(\frac{\mu}{c_1}\right)^{\frac{1}{\rho+\gamma-1}}.
    \end{equation}

    Take $i=n$ in Eq. \eqref{eq:first_order_117} we have 
    \begin{align}\notag
        0&= \frac{\partial u_n}{\partial x_n}\Big|_{x_n=x_n^*} \\ \notag
        &=\frac{\mu}{s^{\gamma}+\alpha s^{\beta}} -\frac{\mu x_n (\gamma s^{\gamma-1}+\alpha\beta s^{\beta-1})}{(s^{\gamma}+\alpha s^{\beta})^2} -\rho c_n x_n^{\rho-1} \\ \notag
        &\geq \frac{ \mu}{s^{\gamma}+\alpha s^{\beta}} -\frac{\mu s( s^{\gamma-1}+\alpha s^{\beta-1})}{n(s^{\gamma}+\alpha s^{\beta})^2}  -\rho c_n x_n^{\rho-1} &\mbox{since $\gamma,\beta\leq 1$}\\ \notag
        &\geq \frac{ \mu}{s^{\gamma}+\alpha s^{\beta}}\left(1-\frac{1}{n}\right) -\frac{\rho c_n \cdot 2^{\frac{(\rho-1)^2}{\rho+\gamma-1}}}{n^{\rho-1}}\cdot\left(\frac{\mu}{c_1}\right)^{\frac{\rho-1}{\rho+\gamma-1}} &\mbox{from Eq. \eqref{eq:166}}\\ \notag
        & \geq \frac{\mu}{s^{\gamma}(1+\alpha)}\left(1-\frac{1}{n}\right) -\frac{\rho c_n\cdot2^{\frac{(\rho-1)^2}{\rho+\gamma-1}}}{n^{\rho-1}c_1^{\frac{\rho-1}{\rho+\gamma-1}}}\cdot\mu^{\frac{\rho-1}{\rho+\gamma-1}} &\mbox{since $\beta\leq\gamma\leq 1$} \\\label{eq:177}
        & \geq \left[\frac{2^{\frac{-\gamma(\rho-1)}{\rho+\gamma-1}}c_n^{\frac{\gamma}{\rho+\gamma-1}}}{1+\alpha}\left(1-\frac{1}{n}\right) - \frac{\rho c_n\cdot2^{\frac{(\rho-1)^2}{\rho+\gamma-1}}}{n^{\rho-1}c_1^{\frac{\rho-1}{\rho+\gamma-1}}}\right] \cdot\mu^{\frac{\rho-1}{\rho+\gamma-1}} &\mbox{since $s\leq2^{\frac{\rho-1}{\rho+\gamma-1}}\cdot\left(\frac{\mu}{c_1}\right)^{\frac{1}{\rho+\gamma-1}}$} \\ \label{eq:179}
        &>0.
    \end{align}
    The last Eq. \eqref{eq:179} holds because when $n\rightarrow +\infty$, $1/n^{\rho-1}\rightarrow 0$ and therefore the expression in the brackets of Eq. \eqref{eq:177} is strictly positive when $n$ is large enough. Eq. \eqref{eq:179} draws a contradiction to the first-order equilibrium condition which completes the proof.

\end{proof}

Now we are ready to prove Theorem \ref{thm:asym_s}. Recall that $\x^*=(x^*_1,\cdots, x^*_n)$ is the unique PNE of the game \basegame{}$(\alpha,\beta,\gamma,\mu,\rho, \{c_i\}_{i=1}^n)$, and $s^*=\sum_{i=1}^n x^*_i$ is the total content creation among all creators at the PNE. First of all, we argue that $(0,\cdots,0)$ is not an equilibrium because when $\mu$ is larger than $\max_{i=1}^n\{c_i\}$, we have $\frac{\partial u_i}{\partial x_i}|_{x_i=0}>0$, meaning each creator $i$ can increase $x_i$ and gain a higher utility. To simplify the notations in the following analysis, we omit the superscript $*$ when there is no ambiguity. 
\begin{proof}\textbf{[of Theorem \ref{thm:asym_s}]}

From the first-order condition Eq. \eqref{eq:first_order_117}

\begin{equation}
    \frac{\partial u_i}{\partial x_i}=\frac{\mu}{s^{\gamma}+\alpha s^{\beta}} -\frac{\mu x_i(\gamma s^{\gamma-1}+\alpha\beta s^{\beta-1})}{(s^{\gamma}+\alpha s^{\beta})^2} -c_i\rho x_i^{\rho-1}=0, \forall i\in[n],
\end{equation}
we have 
\begin{equation}\label{eq:332}
    \frac{\mu}{s^{\gamma}+\alpha s^{\beta}}-c_i\rho x_i^{\rho-1}>0.
\end{equation}
On the other hand, from Lemma \ref{lm:max_x} and $\gamma,\beta\leq 1$ we have 
\begin{align}\notag
    0 &= \frac{\mu}{s^{\gamma}+\alpha s^{\beta}} -\frac{\mu x_i(\gamma s^{\gamma-1}+\alpha\beta s^{\beta-1})}{(s^{\gamma}+\alpha s^{\beta})^2} -c_i\rho x_i^{\rho-1} \\ \notag
     &> \frac{\mu}{s^{\gamma}+\alpha s^{\beta}} -\frac{\mu s(s^{\gamma-1}+\alpha s^{\beta-1})}{2(s^{\gamma}+\alpha s^{\beta})^2} -c_i\rho x_i^{\rho-1} \\\label{eq:204}
     & = \frac{\mu}{2(s^{\gamma}+\alpha s^{\beta})} -c_i\rho x_i^{\rho-1}.
\end{align}

Combining Eq. \eqref{eq:332} and \eqref{eq:204} we have

\begin{equation}\label{eq:210}
   \left(\frac{\mu}{2c_i\rho(s^{\gamma}+\alpha s^{\beta})}\right)^{1/(\rho-1)} < x_i <\left(\frac{\mu}{c_i\rho(s^{\gamma}+\alpha s^{\beta})}\right)^{1/(\rho-1)}.
\end{equation}
Taking summation over both sides of Eq. \eqref{eq:210} for $1\leq i\leq n$ and then taking exponential to $(\rho-1)$ we obtain 
\begin{equation}\label{eq:214}
   \frac{\mu}{2\rho(s^{\gamma}+\alpha s^{\beta})} < \frac{s^{\rho-1}}{\|\c^{-1}\|_{1/(\rho-1)}} <\frac{\mu}{\rho(s^{\gamma}+\alpha s^{\beta})},
\end{equation}
where $\|\c^{-1}\|_{1/(\rho-1)}=\left(\sum_{i=1}^n c_i^{-1/(\rho-1)}\right)^{\rho-1}$ is the $L_{1/(\rho-1)}$-norm of $\c^{-1}=(c_1^{-1},\cdots,c_n^{-1})$. Hence, we obtain 
\begin{equation*}
    \frac{\mu}{2\rho}\|\c^{-1}\|_{1/(\rho-1)}<s^{\gamma+\rho-1}+\alpha s^{\beta+\rho-1}<\frac{\mu}{\rho}\|\c^{-1}\|_{1/(\rho-1)}.
\end{equation*}
When $\beta\leq\gamma$, we also have 
\begin{equation*}
    s^{\gamma+\rho-1}<s^{\gamma+\rho-1}+\alpha s^{\beta+\rho-1}<(1+\alpha)s^{\gamma+\rho-1}.
\end{equation*}
Therefore, we conclude that 

% \begin{equation}\label{eq:415}
% (2(1+\alpha)\rho)^{-\frac{1}{\gamma+\rho-1}}\mu^{\frac{1}{\gamma+\rho-1}} \|\c^{-1}\|_{\frac{1}{\rho-1}}^{\frac{1}{\gamma+\rho-1}}<s < \rho^{-\frac{1}{\gamma+\rho-1}}\mu^{\frac{1}{\gamma+\rho-1}} \|\c^{-1}\|_{\frac{1}{\rho-1}}^{\frac{1}{\gamma+\rho-1}},
% \end{equation}

\begin{equation}\label{eq:415}
(2(1+\alpha)\rho)^{-1}\mu \|\c^{-1}\|_{\frac{1}{\rho-1}}<s^{\gamma+\rho-1} < \rho^{-1}\mu\|\c^{-1}\|_{\frac{1}{\rho-1}},
\end{equation}
which yields Eq. \eqref{eq:s_order} with the constant $C_{\rho}=\rho^{-1}$.

\end{proof}

\section{Proof of Theorem \ref{thm:non_exist_PNE}}\label{app:non_exist_PNE}

\begin{proof}
Consider an \ingame{} instance \ingame{}$(\bm{\alpha},\bm{\beta},\bm{\gamma},\bm{\mu},\{c_i\}_{i=1}^n)$ with $n=K=2$, where $\bm{\alpha}=(0.25, 0.25), \bm{\beta}=(0.5, 0.5), \bm{\gamma}=(1.0,1.0), \bm{\mu}=(3.0, 2.0)$, and the cost functions for two players are given by
$$c_1(\x_1)=7(x_{11}+x_{12})^2, c_2(\x_2)=7(x_{21}+x_{22})^2.$$
Next we show that the PNE of \ingame{}$(\bm{\alpha},\bm{\beta},\bm{\gamma},\bm{\mu},\{c_i\}_{i=1}^n)$ does not exist. Suppose it has an PNE $\y=(\y_1, \y_2)$, we enumerate all the possibilities and draw contradictions accordingly. First of all, by definition $(\bot,\bot)$ cannot be an PNE as both players get negative utility at $(\bot,\bot)$. In this case, either player can change her strategy to $\y_i=(0, 0)$ and increase her utility to zero. Since our constructed instance is symmetric, we only need to further exclude the possibility of the following two types of PNEs:
\begin{enumerate}
    \item If $\y_1\neq\bot, \y_2\neq \bot$, \ingame{} degenerates to \exgame{} with the same parameters but $\bm{\alpha}=(0, 0)$. According to Theorem \ref{thm:monotone_G}, such \exgame{} instance has a unique PNE. Using Algorithm \ref{alg:mmd} we can computationally pin down its PNE: $\y_1=(x_{11},x_{12})=\y_2=(x_{21},x_{22})=(0.179, 0.120)$. And the utilities for both players at $(\y_1,\y_2)$ is 
    $$u_1(\y_1,\y_2)=\frac{3x_{11}}{x_{11}+x_{21}}+\frac{2x_{12}}{x_{12}+x_{22}}-7(x_{11}+x_{12})^{2}=1.875.$$ 
    If we let the first player change her strategy from $(0.179, 0.120)$ to $\bot$, her utility will be 
    $$u_1(\bot,\y_2)=\frac{3\times0.25x_{21}^{0.5}}{x_{21}+0.25x_{21}^{0.5}}+\frac{2\times0.25x_{22}^{0.5}}{x_{22}+0.25x_{22}^{0.5}}=1.953>1.875=u_1(\y_1,\y_2).$$
    As a result, $((0.179, 0.120), (0.179, 0.120))$ is not a PNE of \ingame{} and thus \ingame{} cannot have any PNE satisfying $\y_1\neq\bot, \y_2\neq \bot$.
    \item if \ingame{} has a PNE of the form $(\bot, (x_{21}, x_{22}))$, then we have 
    \begin{align*}
       (x_{21}, x_{22}) &=\arg\max_{\y_2}\{u_2(\y_2,\bot)\} \\
    & =\arg\max_{(x_{21}, x_{22})\in\RR_{\geq 0}^2}\{\frac{x_{21}\mu_1}{x_{21}+\alpha x_{21}^{\beta}}+\frac{x_{22}\mu_2}{x_{22}+\alpha x_{22}^{\beta}}-c_2(x_{21}+x_{22})^{\rho}\}. 
    \end{align*}
    By plugging in the game parameters and solving the RHS convex optimization problem, we obtain $(x_{21}, x_{22})=(0.124, 0.085)$.
    Solving the following convex optimization we have
    $$\arg\max_{\y_1}\{u_1(\y_1, (x_{21}, x_{22}))\}=(0.176, 0.118),$$ and we can verify that
    $$u_1((0.176, 0.118), (0.124, 0.085))=2.316>2.168=u_1(\bot, (0.124, 0.085)).$$
    Hence, player $1$ would change her strategy from $\bot$ to $(0.176, 0.118)$. Then we can verify given player $1$'s strategy $(0.176, 0.118)$, player $2$'s best response excluding $\bot$ is $(0.179, 0.120)$. However, 
    $$u_2((0.179, 0.120), (0.176, 0.118))=1.894<1.962=u_2(\bot, (0.176, 0.118)),$$ 
    which means conditioned on player $1$'s strategy $(0.176, 0.118)$, player $2$ would switch to $\bot$. Hence, starting from any PNE of the form $(\bot, (x_{21}, x_{22}))$, we have the following best-response chain:
    $$(\bot, (x_{21}, x_{22}))\rightarrow ((0.176, 0.118), (x_{21}, x_{22}))\rightarrow((0.176, 0.118), \bot).$$
    Denote $(0.176, 0.118)=(x^*_{11},x^*_{12})$ and $(0.124, 0.085)=(x^*_{21},x^*_{22})$. Since the constructed \ingame{} instance is symmetric, starting from $((x^*_{11},x^*_{12}), \bot)$ we also have the following best-response chain
    $$( (x^*_{21}, x^*_{22}), \bot)\rightarrow ((x^*_{21}, x^*_{22}), (x^*_{11},x^*_{12}))\rightarrow(\bot, (x^*_{11},x^*_{12}))\rightarrow (\bot,(x^*_{21},x^*_{22})).$$
    Therefore, starting from any PNE $(\bot,(x^*_{21},x^*_{22}))$, an alternative best response update from both players would form the following loop
    \begin{align*}
         (\bot, (x^*_{21}, x^*_{22}))&\rightarrow ((x^*_{11},x^*_{12}), (x^*_{21}, x^*_{22}))\rightarrow((x^*_{11},x^*_{12}), \bot)\rightarrow ( (x^*_{21}, x^*_{22}), \bot) \\ & \rightarrow ((x^*_{21}, x^*_{22}), (x^*_{11},x^*_{12}))\rightarrow(\bot, (x^*_{11},x^*_{12}))\rightarrow (\bot,(x^*_{21},x^*_{22})).
    \end{align*}
    Hence, any PNE of the form $(\bot,(x_{21},x_{22}))$ does not exist.
    \item By symmetry, any PNE of the form $((x_{11},x_{12}), \bot)$ does not exist as well.
\end{enumerate}
Therefore, we conclude that the PNE of our constructed instance \ingame{} does not exist. 
\end{proof}

Our example shows that the PNE of \ingame{} need not exist even when $n=K=2$ and $c_i$ are strongly convex functions. We should note that this example can be easily extend to \ingame with arbitrarily large $n$. To see this, consider an instance \ingame{}$(\bm{\alpha},\bm{\beta},\bm{\gamma},\bm{\mu},\{c_i\}_{i=1}^n)$ with $K=2, n>2$, where $\bm{\alpha}=(0.25/(n-1), 0.25/(n-1)), \bm{\beta}=(0.5, 0.5), \bm{\gamma}=(1.0,1.0), \bm{\mu}=(3.0, 2.0)$, and the cost functions for the first two players are 
$$c_1(\x_1)=7(x_{11}+x_{12})^2, c_2(\x_2)=7(x_{21}+x_{22})^2,$$
while the cost for the remaining players are 
$$c_1(\x_1)=M(x_{i1}+x_{i2})^2, i\geq 3,$$
where $M>0$ is a large number. In this game, we can choose sufficiently large $M$ such that as long as there are still human players in the game, $\bot$ is the best strategy of player $i$ for any $i\geq 3$. In this case, conditioned on the other $n-2$ players' strategies, the sub-game of the first two players is the same as the 2-player counterexample we have shown and we can similarly identify the best-response loop starting from any potential PNE.

\section{Proof of Theorem \ref{thm:PNE_extend}} \label{app:proof_PNE_extend}
The formal proof of Theorem \ref{thm:PNE_extend} relies on Proposition \ref{lm:cost_bound}, Lemma \ref{lm:order_c} and Theorem \ref{thm:asym_s}. Before diving into the details, we point out a simple fact: for any \ingamed{}$(\alpha,\beta,\gamma,\mu,\rho,\{c_i\})$, if a group of players commits to strategy $\bot$ while the remaining players (of size $m$) decide not to use $\bot$, \ingamed{} degrades to a standard game \basegame{}$(\alpha',\beta,\gamma,\mu,\rho,\{c_i\})$ with the new parameter $\alpha'=\alpha(n-m)$ and therefore still admits a unique PNE. In the following analysis, we will use $y_i$ to refer to any strategy in $\Y_i=\X_i\cup \{\bot\}$, and use $z_i$ to refer to player $i$'s human strategy in $\X_i\subset \RR_{\geq 0}$, and reserve the notation $x_i$ to denote the PNE strategy of player $i$ given that the identities of players who use $\bot$ are known. 

\begin{proof}\textbf{[of Theorem \ref{thm:PNE_extend}]}
    
The proof is organized by a mathematical induction argument. We will prove the following claim:

\vspace{2mm}
\noindent
{\bf The Induction Claim:} For $k=0, 1, \cdots, n-1$, if $(x_1, \cdots, x_{n-k}, \bot,\cdots,\bot)$ is not an PNE, then the $(n-k)$-th player can unilaterally deviate her strategy to $\bot$ to increase her utility. 

\vspace{2mm}
\noindent
{\bf The Base Case:} When $k=0$, we need to show that if $\y=(x_1,\cdots, x_n)$ is not an PNE of \ingamed{}, then the $n$-th player can switch to $\bot$ to increase her utility. Since $(x_1,\cdots, x_n)$ is not an PNE, there exists a player $j\in [n]$ such that her utility strictly increases when switching to $\bot$. If $j=n$, our claim is true; otherwise, for player $j$ it holds that $u_j(\bot, \x_{-j})>u_j(x_j, \x_{-j})$, i.e.,
\begin{equation}\label{eq:464}
    \frac{\alpha (s-x_j)^{\beta-\gamma+1}}{(s-x_j)^{\gamma}+ \alpha (s-x_j)^{\beta} }\cdot \mu  > \frac{x_j}{s^{\gamma} }\cdot \mu - c_jx_j^{\rho},
\end{equation}
where $s=\sum_{i=1}^n x_i$.

Now we show that if Eq. \eqref{eq:464} holds for some $j<n$, it also holds for $j=n$ and thus the $n$-th player's best response is $\bot$. To see this, note that the function $f(s)=\frac{s^{\beta-\gamma+1}}{s^{\gamma}+ \alpha s^{\beta} }=\frac{1}{s^{2\gamma-\beta-1}+\alpha s^{\gamma-1}}$ is non-decreasing in $s$, because $\gamma-1\leq 0$ and from $\tilde{\beta}+\tilde{\gamma}\geq 1$ we have $\gamma-1\leq 0, 2\gamma-\beta-1\leq 0$. On the other hand, from Theorem \ref{lm:order_c} we also have $u_j(x_j,\x_{-j})\geq u_n(x_n,\x_{-n})$ and $x_j\geq x_n$. As a result, we obtain

\begin{equation*}
   \frac{\alpha (s-x_n)^{\beta-\gamma+1}}{(s-x_n)^{\gamma}+ \alpha (s-x_n)^{\beta} }\cdot \mu  \geq \frac{\alpha (s-x_j)^{\beta-\gamma+1}}{(s-x_j)^{\gamma}+ \alpha (s-x_j)^{\beta} }\cdot \mu   > \frac{x_j}{s^{\gamma} }\cdot \mu - c_jx_j^{\rho} \geq \frac{x_n}{s^{\gamma} }\cdot \mu - c_nx_n^{\rho},
\end{equation*}
which suggests $u_n(\bot, \x_{-n})>u_n(x_n, \x_{-n})$.

\vspace{2mm}

\noindent
{\bf The Induction Argument:}
Suppose the induction claim holds for any $i\leq k-1$. Consider the case $i=k$, we only need to show that if $\y=(x_1, \cdots, x_{n-k}, \bot,\cdots,\bot)$ is not an PNE, the following two statements holds:
\begin{enumerate}
    \item {\bf Statement-1}: If there exists a player $j\leq n-k$ who can deviate to $\bot$ to increase her utility, then player $n-k$ can also deviate to $\bot$ to increase her utility.
    \item {\bf Statement-2}: For any $j\geq n-k+1$, player $j$ does not want to deviate from $\bot$ to any $z_j\in \X_j$ as it will always decrease her utility.
\end{enumerate}
First we prove statement-1. $s=\sum_{i=1}^{n-k}x_i$. For such a player $j$, it holds that $u_j(\bot, \y_{-j})>u_j(x_j, \y_{-j})$, i.e.,
\begin{equation}\label{eq:485}
    \frac{\alpha (s-x_j)^{\beta-\gamma+1}}{(s-x_j)^{\gamma}+ \alpha (k+1) (s-x_j)^{\beta} }\cdot \mu   > \frac{x_j}{s^{\gamma}+ \alpha k s^{\beta} }\cdot \mu - c_jx_j^{\rho}.
\end{equation}
Now we can see that the argument we want to make here is exactly the same as the one in the base case if we replace the parameter $\alpha$ with $\alpha(k+1)$ in the base case. Therefore, we can also show that Eq. \eqref{eq:485} holds for $j=n-k$. Hence, statement-1 is true.

Now we prove Statement-2. We will show that if there exists $j\geq n-k+1$ who can deviate from $\bot$ to some $z_j\in\X_j$ to increase her utility, then at $\y'=(x'_1,\cdots,x'_{n-k},x'_{n-k+1},\bot,\cdots,\bot)$, the $(n-k+1)$-th player would not have switched to $\bot$ and therefore contradicting the induction claim for $i=k-1$. Here we use $(x'_1,\cdots,x'_{n-k},x'_{n-k+1})$ to denote the human players' PNE given that the remaining $(k-1)$ players are using $\bot$. For such a player $j$ we have $u_j(\bot, \y_{-j})<\max_{z_j\in\X_j}u_j(z_j, \y_{-j})$, i.e.,
\begin{equation}\label{eq:491}
    \frac{\alpha s^{\beta-\gamma+1}}{s^{\gamma}+ \alpha k s^{\beta} }\cdot \mu   < \max_{z_j}\left\{\frac{z_j}{(s+z_j)^{\gamma}+ \alpha (k-1) (s+z_j)^{\beta} }\cdot \mu - c_jz_j^{\rho}\right\}.
\end{equation}
By the induction claim we also know at $\y'=(x'_1,\cdots,x'_{n-k},x'_{n-k+1},\bot,\cdots,\bot)$, the $(n-k+1)$-th player can deviate to $\bot$ to increase her utility, which means 
\begin{equation}\label{eq:495}
    \frac{\alpha (s')^{\beta-\gamma+1}}{(s')^{\gamma}+ \alpha k (s')^{\beta} }\cdot \mu   > \frac{x'_{n-k+1}}{(s'+x'_{n-k+1})^{\gamma}+ \alpha (k-1) (s'+x'_{n-k+1})^{\beta} }\cdot \mu - c_{n-k+1}(x'_{n-k+1})^{\rho},
\end{equation}
where $s'=\sum_{i=1}^{n-k}x'_i$. From Lemma \ref{lm:order_c} we have $s'<s$. On the one hand, by the definition of PNE we have
\begin{align*}
    \text{RHS of Eq. \eqref{eq:495}}&=\max_{z}\left\{\frac{z}{(s'+z)^{\gamma}+ \alpha (k-1) (s'+z)^{\beta} }\cdot \mu - c_{n-k+1}z^{\rho}\right\} \\
    &>  \max_{z}\left\{\frac{z}{(s+z)^{\gamma}+ \alpha (k-1) (s+z)^{\beta} }\cdot \mu - c_{n-k+1}z^{\rho}\right\} &\mbox{since $s'<s$}\\
    &\geq  \max_{z}\left\{\frac{z}{(s+z)^{\gamma}+ \alpha (k-1) (s+z)^{\beta} }\cdot \mu - c_{j}z^{\rho}\right\} &\mbox{since $c_{n-k+1}\leq c_j$}\\
    &= \text{RHS of Eq. \eqref{eq:491}}.
\end{align*}
On the other hand, because $f(s)=\frac{s^{\beta-\gamma+1}}{s^{\gamma}+ \alpha s^{\beta} }=\frac{1}{s^{2\gamma-\beta-1}+\alpha s^{\gamma-1}}$ is non-decreasing in $s$, when $\tilde{\beta}+\tilde{\gamma}\geq 1$, we also have 
\begin{align*}
    \text{LHS of Eq. \eqref{eq:495}} \leq  \text{LHS of Eq. \eqref{eq:491}}.
\end{align*}
As a result, we obtain
$$\text{LHS of Eq. \eqref{eq:491}}\geq \text{LHS of Eq. \eqref{eq:495}} > \text{RHS of Eq. \eqref{eq:495}} > \text{RHS of Eq. \eqref{eq:491}}, $$
which contradicts Eq. \eqref{eq:491}. Hence, statement-2 holds and   the induction argument is completed.

The induction argument implies that either there exists an $m(\leq n)$ such that the strategy profile $(x_1, \cdots, x_{n-m}, \bot,\cdots,\bot)$ is an PNE of \ingamed{}, or the remaining human player with the highest cost can always switch to $\bot$ to increase her utility until all players will finally switch to $\bot$. Clearly, the second scenario cannot happen because when there is only one human player left in the game, switching to $\bot$ is not her best response as it renders a zero utility. Therefore, we prove that \ingamed{} must possess a PNE with the form $(x_1, \cdots, x_{n-m}, \bot,\cdots,\bot)$. 

Finally we derive the condition for the threshold $m$. According to the induction argument, $(x_1, \cdots, x_{n-m}, \bot,\cdots,\bot)$ is a PNE if and only if:
\begin{enumerate}
    \item at $\y'=(x'_1,\cdots,x'_{n-m},x'_{n-m+1},\bot,\cdots,\bot)$, the $(n-m+1)$-th player can increase her utility when switching to $\bot$;
    \item at $\y=(x_1,\cdots,x_{n-m},\bot,\cdots,\bot)$, the $(n-m)$-th player cannot increase her utility when switching to $\bot$.
\end{enumerate}
These two conditions translate to

\begin{equation}\label{eq:517}
    \frac{\alpha (s')^{\beta-\gamma+1}}{(s')^{\gamma}+ \alpha m (s')^{\beta} }\cdot \mu   \geq \frac{x'_{n-m+1}}{(s'+x'_{n-m+1})^{\gamma}+ \alpha (m-1) (s'+x'_{n-m+1})^{\beta} }\cdot \mu - c_{n-m+1}(x'_{n-m+1})^{\rho},
\end{equation}

\begin{equation}\label{eq:521}
    \frac{\alpha (s-x_{n-m})^{\beta-\gamma+1}}{(s-x_{n-m})^{\gamma}+ \alpha (m+1)(s-x_{n-m})^{\beta} }\cdot \mu   \leq \frac{x_{n-m}}{s^{\gamma}+ \alpha m s^{\beta} }\cdot \mu - c_{n-m}(x_{n-m})^{\rho}.
\end{equation}

% Form Eq. \eqref{eq:stronger_lm_cost} in Theorem \ref{thm:asym_s}, we have 

% By applying Lemma \ref{lm:cost_bound}, Eq. \eqref{eq:517} implies

% \begin{align}\label{eq:543}
%     \frac{\alpha (s')^{\beta-\gamma+1}}{(s')^{\gamma}+ \alpha m (s')^{\beta} }\cdot \mu &\geq \frac{x'_{n-m+1}}{(s'+x'_{n-m+1})^{\gamma}+ \alpha (m-1) (s'+x'_{n-m+1})^{\beta} }\cdot \mu \cdot\left(1-\frac{1}{\rho}\right) \\ \label{eq:544}
%     & > \frac{(1-\frac{1}{\rho})\mu x'_{n-m+1}}{(s'+s'/(n-m))^{\gamma}+ \alpha m (s'+s'/(n-m))^{\beta} } \\ 
%     & > \frac{(1-\frac{1}{\rho})\mu x'_{n-m+1}}{(s')^{\gamma}+ \alpha m (s')^{\beta} } \cdot \frac{1}{[1+1/(n-m)]^{\gamma}}
% \end{align}
% where Eq. \eqref{eq:544} holds because Lemma \ref{lm:order_c} guarantees $x'_{n-m+1}\leq \frac{s'}{n-m}$, and \fan{the upper bound of $m$ is hard to derive from this equation.}

From Theorem \ref{thm:asym_s}, Eq. \eqref{eq:521} implies
\begin{align}\label{eq:547}
    & \frac{\alpha (s-x_{n-m})^{\beta-\gamma+1}}{(s-x_{n-m})^{\gamma}+ \alpha (m+1)(s-x_{n-m})^{\beta} }\cdot \mu \notag \\ \leq & \frac{x_{n-m}}{s^{\gamma}+ \alpha m s^{\beta} }\cdot \mu \cdot\left(1-\frac{1}{2\rho}\right)\\\notag
     < & \frac{x_{n-m}}{s^{\gamma}+ \alpha (m+1)s^{\beta} }\cdot \frac{m+1}{m}\cdot \mu \cdot\left(1-\frac{1}{2\rho}\right) \\ \notag
     < &\frac{x_{n-m}}{(s-x_{n-m})^{\gamma}+ \alpha (m+1)(s-x_{n-m})^{\beta} }\cdot \frac{m+1}{m}\cdot \mu \cdot\left(1-\frac{1}{2\rho}\right), 
\end{align}
where Eq. \eqref{eq:547} holds because Theorem \ref{thm:asym_s} tells us $c_{n-m}(x_{n-m})^{\rho}>\frac{1}{2\rho}\cdot\frac{x_{n-m}\mu}{s^{\gamma}+ \alpha m s^{\beta} }$. Hence, we obtain
\begin{equation}\label{eq:563}
    \alpha (s-x_{n-m})^{\beta-\gamma+1}<x_{n-m} \cdot\frac{m+1}{m}\cdot \left(1-\frac{1}{2\rho}\right).
\end{equation}
Theorem \ref{lm:order_c} guarantees  $x_{n-m}\leq \frac{s}{n-m}$, and $s-x_{n-m}\geq \frac{n-m-1}{n-m}\cdot s$. Therefore, Eq. \eqref{eq:563} can be further simplified to 
\begin{equation}\label{eq:567}
    \alpha s^{\beta-\gamma+1}\cdot \left(\frac{n-m-1}{n-m}\right)^{\beta-\gamma+1}<\frac{s}{n-m}\cdot\frac{m+1}{m}\cdot\left(1-\frac{1}{2\rho}\right).
\end{equation}
Since $\beta-\gamma+1<1$, $\frac{n-m}{n-m-1}\leq 2, \frac{m+1}{m}\leq 2$, from Eq. \eqref{eq:567} we can derive
\begin{equation}
    n-m<\frac{2(2\rho-1)}{\alpha \rho} \cdot s^{\gamma-\beta},
\end{equation}
which implies 
\begin{align}\notag
    \frac{m}{n}&>1-\frac{2(2\rho-1)}{\alpha \rho n} \cdot s^{\gamma-\beta} \\ 
    & > 1-\frac{2(2\rho-1)}{\alpha \rho n} \cdot \left(C_0 \mu^{\frac{1}{\gamma+\rho-1}} \|\c_{n-m}^{-1}\|_{\frac{1}{\rho-1}}^{\frac{1}{\gamma+\rho-1}}\right)^{\gamma-\beta} & \mbox{By Theorem \ref{thm:asym_s}} \\ 
    &= 1-\frac{2(2\rho-1)}{\alpha \rho n} \cdot \left(C_0 \mu^{\frac{1}{\gamma+\rho-1}} \cdot \left(\sum_{i=1}^{n-m} c_i^{-1/(\rho-1)}\right)^{\frac{\rho-1}{\gamma+\rho-1}} \right)^{\gamma-\beta}\\ 
    & \geq 1-\frac{2(2\rho-1)}{\alpha \rho n} \cdot \left(C_0 \mu^{\frac{1}{\gamma+\rho-1}} \cdot \left(n c_1^{-1/(\rho-1)}\right)^{\frac{\rho-1}{\gamma+\rho-1}}\right)^{\gamma-\beta} \\ 
    & = 1-C\cdot\frac{\mu^\frac{\gamma-\beta}{\gamma+\rho-1}}{\alpha n^{1-\frac{(\gamma-\beta)(\rho-1)}{\gamma+\rho-1}}},
\end{align}
where $C_0=\rho^{-\frac{1}{\gamma+\rho-1}}$ is a constant depending on $\rho,\gamma$ (From Eq. \eqref{eq:415}), and $C=\frac{2(2\rho-1)}{ \rho}\cdot (\rho c_1)^{-\frac{\gamma-\beta}{\gamma+\rho-1}}$ is a constant depending on $(\beta,\gamma,\rho,c_1)$.
\end{proof}

\section{Equilibrium Solvers Used in Experiments}\label{app:pga_solver}
\subsection{PNE solver for \exgame{}}

\begin{algorithm}[h]
   \caption{Multi-agent Mirror Descent (MMD) with perfect gradient}
   \label{alg:mmd}
\begin{algorithmic}
   \STATE {\bfseries Input:} Maximum iteration number $T$, step size $\eta$, each player $i$'s utility function $u_i$, error tolerance $\epsilon$, initial strategy $\x_i=\x^{(0)}_i$.
   \REPEAT 
        \STATE Compute the exact gradient $\g_i = \nabla_i u_i(\x_i, \x_{-i}), \forall i \in [n]$, 
        \STATE Update $\x_i \leftarrow \text{Proj}_{\X_i}(\x_i+\eta \g_i), \forall i\in [n]$.
   \UNTIL {Maximum iteration number is reached or $\|(\g_1,\cdots,\g_n)\|_2<\epsilon.$}

   \STATE {\bfseries Output: $(\x_1,\cdots,\x_n)$. }
\end{algorithmic}
\end{algorithm}

Since the strategy set $\X_i=[0, +\infty)$, we can simply choose a projection mapping $\text{Proj}_{\X_i}(\x)=(\max(x_i, 0))_{i=1}^n$. The utility functions of \exgame{} is differentiable and has closed forms so we can explicitly implement their gradients. Through our experiment, the default $T=1000, \eta=0.05,\epsilon=1e-4, \x_i^{(0)}=(0.1,\cdots,0.1)$. Algorithm \ref{alg:mmd} is a simplified version of Algorithm 1 in \cite{bravo2018bandit} where we replace the gradient estimation to the exact gradient. According to Theorem 5.1 in \cite{bravo2018bandit}, Algorithm \ref{alg:mmd} converges to the unique PNE of \exgame{} with probability 1. 

\subsection{PNE solver for \ingamed{}}
Since the PNE of \ingamed{} might not be unique, we present two PNE solvers for \ingamed{}. The first Algorithm \ref{alg:pne1} is to pin down the PNE of the form $(x_1,\cdots,x_{n-m},\bot,\cdots,\bot)$ guaranteed by Theorem \ref{thm:PNE_extend}, and the second Algorithm \ref{alg:pne2} is for finding an arbitrary PNE. Both of them use Algorithm \ref{alg:mmd} as a subroutine. For the ease of notation, we denote \basegame{}$(m)$ as the subgame for the first $n-m$ players when the remaining $m$ players with the top-$m$ highest costs are commited to strategy $\bot$.

\begin{algorithm}[h]
   \caption{Solving for a targeted PNE of \ingamed{}}
   \label{alg:pne1}
\begin{algorithmic}
   \STATE {\bfseries Input:} Each player $i$'s utility function $u_i$.
    \STATE {\bfseries Initialization:} Use Algorithm \ref{alg:mmd} to solve $\x^{(0)}=\text{MMD}($\basegame{}$(0))$, and set $\x=\x^{(0)}$.
   \FOR{$i=n$ {\bfseries to} $1$}
        \STATE Compute $u_i(\bot, \x_{-i})$ and $u_i(\x_i, \x_{-i})$. %to determine if player $n$ wants to switch to $\bot$.
        \IF{$u_i(\bot, \x_{-i})>u_i(x_i, \x_{-i})$} 
            \STATE Set $x_i=\bot$.
            \STATE Update $\x_{-i}=\text{MMD}($\basegame{}$(n+1-i))$.
        \ELSE
            \STATE Break.
        \ENDIF
   \ENDFOR

   \STATE {\bfseries Output: $\x$. }
\end{algorithmic}
\end{algorithm}

To find an arbitrary PNE of \ingamed{}, first we need a PNE checker as a subroutine, as shown in the following Algorithm \ref{alg:pne_checker}. Algorithm \ref{alg:pne_checker} takes as input the game instance \ingamed{} and an arbitrary joint strategy profile $\x\in \cup_{i=1}^n\Y_i$. To verify whether $\x$ is an PNE of \ingamed{}, it checks for every player whether their is a better response: for a human player $i$, it simply compares $i$'s current utility and the utility if adopting $\bot$; for a GenAI player $i$, it needs to compare $i$'s current utility and the best possible human strategy, which requires solving an optimization problem with $u_i$ as the objective function. Thanks to the monotonicity of \ingamed{}, we know $u_i$ is concave so this optimization is tractable. If no one would like to deviate, $\x$ passes the checker and Algorithm \ref{alg:pne_checker} returns True and the same $\x$ meaning $\x$ is an PNE; otherwise \ref{alg:pne_checker} returns False and the new $\x$ incorporating some player's best response. Algorithm \ref{alg:pne2} works by shuffling players' indices first, and call the PNE checker \ref{alg:pne_checker} at each iteration to allow an arbitrary player to improve her utility until achieving a PNE. We note that although there is no finite time convergence guarantee for Algorithm \ref{alg:pne2}, but we can be sure as long as it terminates it must return an PNE of \ingamed{}. In our experiments it always converge within 5000 iterations. 

\begin{algorithm}[h]
   \caption{PNE checker for \ingamed{}}
   \label{alg:pne_checker}
\begin{algorithmic}
   \STATE {\bfseries Input:} Each player $i$'s utility function $u_i$ and $\x=(x_1,\cdots,x_n)$.
   \FOR{$i=1$ {\bfseries to} $n$}
        \IF{$x_i\neq \bot$ ~and~ $u_i(\bot, \x_{-i})>u_i(x_i, \x_{-i})$} 
            \STATE Set $x_i=\bot$.
            \STATE Update other human players' strategy by solving the new PNE of players excluding $x_i$.
            \STATE Return False, $\x_i$.
        \ELSIF{$x_i=\bot$ ~and~ $\max_{y_i\in\RR_{\geq 0}} u_i(y_i,\x_{-i})>u_i(\bot, \x_{-i})$}
            \STATE Set $x_i=\arg\max_{y_i\in\RR_{\geq 0}} u_i(y_i,\x_{-i})$.
            \STATE Update other human players' strategy by solving the new PNE of players excluding $x_i$.
            \STATE Return False, $\x_i$.
        \ELSE 
            \STATE Continue;
        \ENDIF
        
   \ENDFOR
   \STATE Return True, $\x_i$.

\end{algorithmic}
\end{algorithm}

\begin{algorithm}[h]
   \caption{Solving for an arbitrary PNE of \ingamed{}}
   \label{alg:pne2}
\begin{algorithmic}
   \STATE {\bfseries Input:} Each player $i$'s utility function $u_i$ and $\x=(x_1,\cdots,x_n)$.
    \STATE {\bfseries Initialization:} Shuffle $(c_1,\cdots,c_n)$, PNE$\_$FLAG=False.
   \WHILE{PNE$\_$FLAG is False}
        \STATE PNE$\_$FLAG, $\x$ = PNE$\_$checker(\ingamed{}, $\x$).
    \ENDWHILE
   \STATE {\bfseries Output: $\x$. }
\end{algorithmic}
\end{algorithm}

%\section{Additional Experiments}\label{app:experiment}